\documentclass[aos,preprint]{imsart}

\usepackage{amsthm,amssymb,amsmath,natbib,array,graphicx,mathrsfs,comment}
\RequirePackage[colorlinks,citecolor=blue,urlcolor=blue]{hyperref}

\arxiv{1707.00833}

\startlocaldefs

\newtheorem{thm}{Theorem}[section]
\newtheorem{lem}[thm]{Lemma}
\newtheorem{prop}[thm]{Proposition}
\newtheorem{clm}{Claim}
\newtheorem{cor}[thm]{Corollary}
\theoremstyle{definition}
\newtheorem{rem}[thm]{Remark}

\newcommand{\R}{\mathbb{R}}
\newcommand{\IER}{\textup{IER}}
\renewcommand{\P}{\mathsf{P}}
\newcommand{\E}{\mathsf{E}}
\newcommand{\Var}{\mathsf{Var}}
\newcommand{\1}{\mathbf{1}}
\newcommand{\e}{\mathbf{e}}
\newcommand{\mhat}{\widehat{\mu}}
\newcommand{\shat}{\widehat{\sigma}}

\newcommand{\Mb}{\mathbb{M}_n}
\newcommand{\Hc}{\mathcal{H}}

\newcommand{\cc}{\delta} 
\renewcommand{\sc}{\rho} 
 
\endlocaldefs

\begin{document}

\begin{frontmatter}

\title{Two-sample Hypothesis Testing for Inhomogeneous Random Graphs}
\runtitle{Two-sample Testing for Random Graphs}


\begin{aug}
\author{\fnms{Debarghya} \snm{Ghoshdastidar}\corref{Corresponding author}\thanksref{t1,a1}\ead[label=e1]{debarghya.ghoshdastidar@uni-tuebingen.de}},
\author{\fnms{Maurilio} \snm{Gutzeit}\thanksref{t2,a2}\ead[label=e2]{maurilio.gutzeit@ovgu.de}},
\author{\fnms{Alexandra} \snm{Carpentier}\thanksref{t3,a2}\ead[label=e3]{alexandra.carpentier@ovgu.de}}
\and
\author{\fnms{Ulrike} \snm{von Luxburg}\thanksref{t4,a1,a3}\ead[label=e4]{luxburg@informatik.uni-tuebingen.de}}

\thankstext{t1}{Partially supported by the Deutsche Forschungsgemeinschaft (DFG, FOR1735), the Institutional Strategy of the University of T{\"u}bingen (DFG, ZUK 63) and by the Baden-W{\"u}rttemberg Stiftung through the Eliteprogramm for Postdocs.}
\thankstext{t2}{Supported by the Deutsche Forschungsgemeinschaft (DFG) Emmy Noether grant MuSyAD (CA 1488/1-1).}
\thankstext{t3}{Partially supported by the DFG Emmy Noether grant MuSyAD (CA 1488/1-1), by the DFG - 314838170, GRK 2297 MathCoRe, by the DFG GRK 2433 DAEDALUS (384950143/GRK2433), by the DFG CRC 1294 `Data Assimilation', Project A03, and by the UFA-DFH through the French-German Doktorandenkolleg CDFA 01-18.}
\thankstext{t4}{Partially supported by the DFG FOR1735, the Institutional Strategy of the University of T{\"u}bingen (DFG, ZUK 63) and DFG Cluster of Excellence ``Machine Learning – New Perspectives for Science'' EXC 2064/1, project number 390727645.}
\runauthor{Ghoshdastidar et al.}

\affiliation{University of T{\"u}bingen\thanksmark{a1}~ Otto von Guericke University Magdeburg\thanksmark{a2}\\
 Max Planck Institute for Intelligent Systems, T{\"u}bingen\thanksmark{a3}}

\address{
D. Ghoshdastidar, U. von Luxburg\\
University of T{\"u}bingen\\
Department of Computer Science\\ 
Maria von Linden Stra{\ss}e 6\\
72076 T{\"u}bingen, Germany\\
\printead{e1}\\
\printead*{e4}\\
M. Gutzeit, A. Carpentier\\
Otto von Guericke University Magdeburg\\
Faculty of Mathematics\\
Universit{\"a}tsplatz 2\\
39106 Magdeburg, Germany\\
\printead{e2}\\
\printead*{e3}\\
}

\end{aug}

\begin{abstract}
The study of networks leads to a wide range of high dimensional inference problems. In many practical applications, one needs to draw inference from one or few large sparse networks. The present paper studies hypothesis testing of graphs in this high-dimensional regime, where the goal is to test between two populations of inhomogeneous random graphs defined on the same set of $n$ vertices. The size of each population $m$ is much smaller than $n$, and can even  be a constant as small as 1. The critical question in this context is whether the problem is solvable for small $m$.

We answer this question from a minimax testing perspective. Let $P,Q$ be the population adjacencies of two sparse inhomogeneous random graph models, and $d$ be a suitably defined distance function. Given a population of $m$ graphs from each model, we derive minimax separation rates for the problem of testing $P=Q$ against $d(P,Q)>\rho$. We observe that if $m$ is small, then the minimax separation is too large  for some popular choices of $d$, including total variation distance between corresponding distributions. This implies that some models that are widely separated in $d$ cannot be distinguished for small $m$, and hence, the testing problem is generally  not solvable in these cases.

We also show that if $m>1$, then the minimax separation is relatively small if $d$ is the Frobenius norm or operator norm distance between $P$ and $Q$. For $m=1$, only the latter distance provides small minimax separation. Thus, for these distances, the problem is solvable for small $m$. We also present near-optimal two-sample tests in both cases, where tests are adaptive with respect to sparsity level of the graphs.

\end{abstract}



\end{frontmatter}


\section{Introduction}

Analysis of random graphs has piqued the curiosity of probabilists since its inception decades ago, but the  widespread use of networks in recent times has made statistical inference from random graphs a topic of immense interest for both theoretical and applied researchers.
This has caused a fruitful interplay between theory and practice leading to deep understanding of statistical problems that, in turn, has led to advancements in applied research.
Significant  progress is clearly visible in problems related to network modelling~\citep{Albert_2002_jour_RevModPhys,Lovasz_2012_book_AMS}, community detection~\citep{Decelle_2011_jour_PhysRevE,Abbe_2016_conf_NIPS}, network dynamics~\citep{Berger_2005_conf_SODA} among others, where statistically guaranteed methods have emerged as effective practical solutions.
Quite surprisingly, the classical problem of hypothesis testing of random graphs is yet to benefit from such joint efforts from theoretical and applied researchers.
It should be noted the problem itself is actively studied in both communities. Testing between brain or `omics' networks have surfaced as a crucial challenge in the context of both modelling and decision making~\citep{Ginestet_2017_jour_AOAS,Hyduke_2013_jour_MolBioSys}.
On the other hand, phase transitions are now known for the problems of detecting high-dimensional geometry or strongly connected groups in large random graphs~\citep{Bubeck_2016_jour_RSA,AriasCastro_2014_jour_AnnStat}. 
However, little progress has been made in the design of consistent tests for general models of random graphs. The present paper takes a step towards addressing this general concern.

While research on testing large random graphs has been limited, hypothesis testing in large dimension is an integral part of modern statistics literature.
In fact, \citet{Bai_1996_jour_StatisticaSinica} demonstrated the need for studying high dimensional statistics through a two-sample testing problem, where one tests between two $n$-variate normal distributions with different means by accessing  $m$ i.i.d. observations from either distributions, where $m\ll n$
(we denote the dimension by $n$ since, in the context of graphs, the number of vertices $n$ governs the dimensionality of the problem).
More recent works in this direction provide tests and asymptotic guarantees as $m\to\infty$ and $\frac{n}{m}\to\infty$~\citep{Chen_2010_jour_AnnStat,Cai_2014_jour_RoyalStatSocB,Ramdas_2015_arxiv_00655}.
Similar studies also exist in the context of testing whether a $n$-dimensional covariance matrix is identity, where only $m\ll n$ i.i.d. observations of the data is available~\citep{Ledoit_2002_jour_AnnStat,Berthet_2013_jour_AnnStat,AriasCastro_2015_jour_Bernoulli}.
It is also known that the computational complexity of this problem is closely related to the clique detection problem in random graphs~\citep{Berthet_2013_jour_AnnStat}.
An extreme version of high dimensional testing arises in the signal detection literature, where one observes a high-dimensional (Gaussian) signal, and tests the nullity of its mean. Subsequently, asymptotics of the problem is considered for $n\to\infty$ but fixed sample size ($m=1$ or a constant), and minimax separation rates between the null and the alternative hypotheses are derived \citep{Ingster_2003_book_Springer,Baraud_2002_jour_Bernoulli,Verzelen_2014_arxiv_1478,Mukherjee_2015_jour_AnnStat}.
The signal detection problem has also been extended in the case of matrices in the context of trace regression~\citep{Carpentier_2015_jour_EJStat} and sub-matrix detection~\citep{Sun_2008_jour_JMLR} among others, where the latter generalises the planted clique detection problem.

In practice, the problem of testing random graphs comes in a wide range of flavours. 
For instance, while dealing with graphs associated with chemical compounds~\citep{Shervashidze_2011_jour_JMLR} or brain networks of several patients collected at multiple laboratories~\citep{Ginestet_2017_jour_AOAS}, one has access to a large number of graphs (large $m$). 
This scenario is more amenable as one can resort to the vast literature of non-parametric hypothesis testing that can even be applied to random graphs.
A direct approach to this problem is to use kernel based tests~\citep{Gretton_2012_jour_JMLR} in conjunction with graph kernels~\citep{Vishwanathan_2010_jour_JMLR,Kondor_2016_conf_NIPS}, which does not require any structural assumptions on the network models.
However, known guarantees for such tests depend crucially on the sample size, and one cannot conclude about the fidelity of such tests for very small $m$.
A more challenging situation arises when $m$ is small, and unfortunately, this is often the case in network analysis. 
For example, if the graphs correspond to brain networks collected from patients in a single lab setup $(m<20)$, brain networks of one individual obtained from test-retest MRI scans~\citep{Landman_2011_jour_Neuroimage},  or molecular interaction networks arising from genomic or proteomic data~\citep{Hyduke_2013_jour_MolBioSys}. Test-retest data of a patient provide only $m=2$ networks, while omics data typically result in one large interaction network, that is, $m=1$.
Hence, designing two-sample tests for small populations for graphs is a problem of immense significance, and yet, practical tests with statistical guarantees are rather limited.

From a theoretical perspective, only a handful of results on testing large random graphs are known. 
While finding hidden cliques have been a long-standing open problem, \citet{AriasCastro_2014_jour_AnnStat,Verzelen_2015_jour_AOAP} for the first time provide a characterisation of the more basic problem of detecting a planted clique in an Erd\H{o}s-R{\'e}nyi  graph. 
\citet{Gao_2017_arxiv_06742} and \citet{Lei_2016_jour_AnnStat} consider generalised variants of this problem while designing tests to distinguish a stochastic block model from an Erd\H{o}s-R{\'e}nyi  graph, or to estimate the number of communities in a stochastic block model, respectively.
In a different direction, \citet{Bubeck_2016_jour_RSA} study the classical problem of testing whether a given graph corresponds to a neighbourhood graph in a high-dimensional space, or it is generated from Erd\H{o}s-R{\'e}nyi model.
This result is in fact a specific instance of the generic problem of detecting whether a network data has a dependence structure or is unstructured~\citep{Ryabko_2017_conf_ALT,Bresler_2018_conf_COLT,Daskalakis_2018_conf_SODA_ising}.
The first study on two-sample testing of graphs, under a relatively broad framework is by~\citet{Tang_2016_jour_JCompGraphStat}, where the authors test between a pair of random dot product graphs which are undirected graphs on a common set of $n$ vertices with mutually independent edges and low-rank population adjacency matrices.
A test statistic based on the difference in adjacency spectral embeddings is shown to be asymptotically consistent as $n\to\infty$ provided that the rank of the population adjacencies is fixed and known.
\citet{Tang_2017_jour_Bernoulli} study a more general problem of comparing two random dot product graphs defined on different vertex sets, where the vertices have latent Euclidean representations. The latent representation can be recovered from the adjacency spectral embedding, and kernel two-sample testing for Euclidean data can be employed to solve the testing problem.
\citet{Ghoshdastidar_2017_conf_COLT} provide a framework to formulate the graph two-sample problem with minimal structural restrictions, and show that this general framework can be used to prove minimax optimality of tests based on triangle counts and spectral properties under special cases.

In this paper, we restrict ourselves to graphs on a common vertex set of size $n$ and sampled from an inhomogeneous Erd\H{o}s-R{\'e}nyi (IER) model~\citep{Bollobas_2007_jour_RSA}, that is, we consider undirected and unweighted random graphs where the edges occur independently, but there is no structural assumption on the population adjacency matrix.
We study the problem of testing between two IER models, where $m$ i.i.d. graphs are observed for each model. 
Apparently, allowing $m\geq1$ appears to be a slight generalisation of the $m=1$ case~\citep{Tang_2016_jour_JCompGraphStat}, but we show that in some situations, the testing problem behaves differently in the $m=1$ and $m>1$ cases.
It is also well established that many graph learning problems have different behaviour in the case of dense and sparse graphs. This is indeed true in the context of testing for geometric structures~\citep{Bubeck_2016_jour_RSA} and community detection~\citep{Verzelen_2015_jour_AOAP}.
Bearing this in mind, we study the two-sample problem at different levels of sparsity of the graph.
A formal description of the problem is presented in Section~\ref{sec_problem}.

Given the above framework, one may resort to a variety of testing procedures. 
A classical approach involves viewing the problem as an instance of closeness testing for high dimensional discrete distributions~\citep{Chan_2014_conf_SODA,Daskalakis_2018_conf_SODA_ising} and using variants of the $\chi^2$-test or related localisation procedures.
On the other hand, exploiting the independence of edges, one may even view the problem as a instance of multiple testing, and may resort to tests based on higher criticism~\citep{Donoho_2004_jour_AnnStat,Donoho_2015_jour_StatScience}.
A more direct approach may be to simply compare the adjacency spectral embeddings~\citep{Tang_2016_jour_JCompGraphStat}, or other network statistics~\citep{Ghoshdastidar_2017_conf_COLT} or even the raw adjacency matrices.
Sections~\ref{sec_notsolvable}--\ref{sec_operator} present a variety of testing problems, which differ in terms of the distance $d(P,Q)$, that is, how we quantify the separation between the two models.
We show that some of the above principles are not useful for small $m$ since the associated testing problems are generally unsolvable in these cases.
However, in some cases, one can construct uniformly consistent tests that work with a small number of observation, even $m=1$.
Section~\ref{sec_discussions} discusses the practicality of graph two-sample testing and also presents minimax separation under special cases of the IER model, such as Erd\H{o}s-R{\'e}nyi or stochastic block models, or under  different notions of sparsity in graphs.
The detailed proofs are provided in the appendix.

\section{Problem statement}
\label{sec_problem}

In this section, we formally state the generic two-sample graph testing problem studied in this paper. 
We also present the minimax framework that forms the basis of our theoretical analysis.

We use the notations $\lesssim$ and $\gtrsim$ to denote the standard inequalities but ignoring absolute constants. Further, we use $\asymp$ to denote that two quantities are same up to possible difference in constant scaling. 
We use $\wedge$ and $\vee$ (or $\bigwedge$ and $\bigvee$) to denote minimum and maximum, respectively.
We also need several standard norms and distances. 
For two discrete distributions, we denote the total variation distance by $TV(\cdot,\cdot)$ and the symmetric Kullback-Leibler (KL) divergence by $SKL(\cdot,\cdot)$. The latter is a symmetrized version of KL-divergence~\citep{Daskalakis_2018_conf_SODA_ising}.
We use the following quantities for any matrix:
\\(i) Frobenius norm, $\Vert \cdot \Vert_F$, is the root of sum of squares of all entries,
\\(ii) max norm, $\Vert\cdot\Vert_{max}$, is largest absolute entry of the matrix,
\\(iii) zero norm, $\Vert\cdot\Vert_0$, is the number of non-zero entries,
\\(iv) operator norm, $\Vert \cdot \Vert_{op}$, is the largest singular value of the matrix, and
\\(v) row sum norm, $\Vert \cdot\Vert_{row}$, (or, the induced $\infty$-norm) is the maximum absolute row sum of the matrix.

\subsection{The model and the testing problem}

Throughout the paper, $V=\{1,2,\ldots,n\}$ denotes a set of $n$ vertices, and we consider undirected graphs defined on $V$. 
Any such graph can be expressed as $G = (V,E_G)$, where $E_G$ is the set of undirected edges.
We use the symmetric matrix $A_G\in\{0,1\}^{n\times n}$ to denote the adjacency matrix of $G$, where $(A_G)_{ij}=1$ if $(i,j)\in E_G$, and 0 otherwise.
The class of inhomogeneous random graphs, or more precisely inhomogeneous Erd\H{o}s-R{\'e}nyi (IER) graphs, on $V$ can be described as follows.
Let $\Mb \subset [0,1]^{n\times n}$ be the set of symmetric matrices with zero diagonal, and off-diagonal entries in $[0,1]$. 
For any $P\in\Mb$, we say that $G$ is an IER graph with population adjacency $P$, denoted by $G\sim \IER(P)$, if the adjacency matrix $A_G$ is a symmetric random matrix such that $(A_G)_{ij} \sim \text{Bernoulli}_{0-1}(P_{ij})$, and  $\left( (A_G)_{ij}\right)_{1\leq i<j\leq n}$ are independent.

Let $P,Q\in \Mb$.
Given $m$ independent observations from each of $\IER(P)$ and $\IER(Q)$,  we would like to test between the alternatives
\begin{equation}
\Hc_0: P = Q
\qquad \text{and} \qquad
\Hc_1: d(P,Q)>\sc
\label{eq_problem_general}
\end{equation}
for some specified distance function $d$ and a threshold $\sc\geq0$.
At this stage, note that the distribution $\IER(P)$ is completely characterised by the expected adjacency matrix $P$.
Hence, $\Hc_0$ in~\eqref{eq_problem_general} is similar both under the mean difference alternative and the general difference alternative~\citep{Ramdas_2015_arxiv_00655}.
Hence, one may assume $d$ to be either a distance between the distributions $\IER(P)$ and $\IER(Q)$, or a matrix distance between $P$ and $Q$.
Different examples of $d$ are considered in Sections~\ref{sec_notsolvable}--\ref{sec_operator}, which result in specific instances of the testing problems.

We note that the complexity of graph inference problems is often governed by the sparsity of the graphs.
To take the effect of sparsity into account, we restrict the problem to models such that $\Vert P\Vert_{max} \vee \Vert Q\Vert_{max} \leq \cc$ for some  $\cc\in(0,1]$ where $\cc$ may decay with $m,n$.
Intuitively, we consider only graphs that are uniformly sparse, that is any edge can occur with probability at most $\cc$.
For instance, if $\cc\asymp\frac1n$, we mostly observe sparse graphs with bounded expected degrees.
Such a uniform sparsity restriction is along the lines of a scalar sparsity parameters introduced in some graph estimation problems~\citep{Klopp_2017_jour_AnnStat}.
More general notions of sparsity may be considered as discussed later in Section~\ref{sec_discussions}.
Based on the above considerations, we formally state the following general framework for graph two-sample testing:
\begin{equation}
 \Hc_0: (P,Q) \in \Omega_0 \qquad \text{vs.} \qquad \Hc_1:  (P,Q)\in\Omega_1,
 \label{eqn_prob}
\end{equation}  
where 
\begin{equation}
\begin{aligned}
\Omega_0 &= \{ (P,Q) \in \Mb\times\Mb : P=Q, ~\Vert P\Vert_{max} \leq \cc\}, 
\\
\Omega_1 &= \{ (P,Q) \in \Mb\times\Mb : d(P,Q) > \sc, ~\Vert P\Vert_{max} \vee \Vert Q\Vert_{max} \leq \cc\}.
\end{aligned}
\label{eqn_prob_sets}
\end{equation}
Note that the hypotheses are governed by the distance $d$, the integers $m,n$, and the positive scalars $\sc,\cc$, where the last two terms may depend on $m,n$.

\subsection{Minimax framework}
Given the graphs $G_1,\ldots,G_m\sim_{\text{iid}} \IER(P)$ and $H_1,\ldots,H_m\sim_{\text{iid}} \IER(Q)$,
a test $\Psi$ is a binary function of $2m$ adjacency matrices, where $\Psi=0$ when the test accepts $\Hc_0$, and $\Psi=1$ otherwise.
The maximum or worst-case risk of a test is given by
\begin{equation*}
R(\Psi,n,m,d,\sc,\cc) =  \sup_{\theta\in \Omega_0} \P_\theta(\Psi = 1) +  \sup_{\theta\in \Omega_1} \P_\theta(\Psi = 0),
\end{equation*}
which is the sum of maximum possible Type-I and Type-II error rates incurred by the test.
Here, we use $\theta$ to denote any tuple $(P,Q)$.
The minimax risk for the problem in~\eqref{eqn_prob}--\eqref{eqn_prob_sets} is defined as
\begin{equation}
R^*(n,m,d,\sc,\cc) = \inf\limits_\Psi R(\Psi,n,m,d,\sc,\cc).
\label{eqn_minimaxrisk}
\end{equation}

Our aim in this paper is to find the minimax separation $\sc^*$ for a given problem, which is the smallest possible $\sc$ such that $R^*(n,m,d,\sc,\cc)\leq \eta$ for some pre-specified $\eta\in(0,1)$.
In the subsequent sections, we consider testing problems where the separation between $P$ and $Q$ is defined in terms of various distance functions. 
We provide bounds for $\sc^*$ for the different testing problems in terms of the  various parameters of the problem~\eqref{eqn_prob}--\eqref{eqn_prob_sets}. 
Though our formal bounds are explicit in terms of $\eta$, we generally assume $\eta$ to be a pre-specified constant (for example, $\eta=0.05$) and focus on the dependence of $\sc^*$ on $n,m,$ and $\cc$.
Our aim is to provide upper and lower bounds for $\sc^*$ that are same up to difference in absolute constants and functions of $\eta$.

\section{Challenges of testing with small sample size}
\label{sec_notsolvable}

The theme of this paper is to test between two populations of sparse graphs, where the sample size $m$ is much smaller than the number of vertices $n$.
Our main interest is in cases where $m$ is a small constant, or may grow very slowly with $n$. 
In this section, we show that in case of some popular distance functions, the testing problem is nearly unsolvable if $m$ is small.

We formalise the notion of unsolvability in the following way. 
For any instance of the two-sample testing problem~\eqref{eqn_prob}--\eqref{eqn_prob_sets}, there is a trivial upper bound for $\sc^*$ which is the maximal possible value that can be attained by $d(P,Q)$ (diameter with respect to $d$). 
As an example, if $d(P,Q) = TV\big(\IER(P),\IER(Q)\big)$, then $\sc^*\leq 1$ trivially. Similarly, if $d(P,Q) = \Vert P-Q\Vert_F$, a trivial upper bound is $\sc^*\leq n\cc$.
On the other hand, for small $m$, if there is a lower bound $\sc_\ell \leq \sc^*$ such that $\sc_\ell$ is equal or close to the trivial upper bound, then there exist model pairs such that $d(P,Q)$ is nearly as large as the diameter and yet cannot be distinguished for small $m$.
Hence, we may conclude that the problem~\eqref{eqn_prob}--\eqref{eqn_prob_sets} with the specific choice of $d$ is unsolvable for small $m$ under a worst-case (minimax) analysis.

We present the first instance of such an impossibility result for the case of total variation distance.
For any two probability mass functions $p$ and $q$, defined on the space of undirected $n$-vertex graphs,
\begin{align*}
    TV(p,q) = \frac12 \sum_G | p(G) - q(G)|,
\end{align*}
where the summation is over all unweighted undirected graphs on $n$ vertices.
We present the following minimax rate in the small sample regime.
\begin{prop}[$\sc^*$ for total variation distance]
\label{prop_tv}
Consider the problem in \eqref{eqn_prob}--\eqref{eqn_prob_sets} with $d(P,Q) = TV\big(\IER(P),\IER(Q)\big)$ and $\cc\in(0,1)$. Let $\eta\in(0,1)$ be the allowable risk. For any $\cc \geq C'\frac{\ln n}{n^2}$, 
\begin{displaymath}
1-\frac1n \leq \sc^* \leq 1
\qquad\text{for~ } m\leq \displaystyle C\sqrt{\ln(1+4(1-\eta)^2)} \frac{n}{\ln n} \;,
\end{displaymath}
where $C\in(0,1)$ and $C'>1$ are absolute constants.

In particular, for large $n$ and $m \lesssim \frac{n}{\ln n}$, we have $\sc^*\approx 1$.
\end{prop}

\begin{proof}[Proof sketch]
    The main idea is to find an appropriate choice of $(P,Q)$ such that $TV\big(\IER(P),\IER(Q)\big) \geq 1 - \frac1n$, and yet they cannot be distinguished using $m\lesssim \frac{n}{\ln n}$ samples. 
    For this, we use standard approaches to derive minimax lower bounds~\citep{Baraud_2002_jour_Bernoulli}. 
    This is described in the present context of two-sample testing of $\IER$ graphs in the appendix.
    
    In particular for the present proof, the choice of $P,Q$ is the following. 
    Under $\Hc_0$, we set $P=Q$ such that the model corresponds to Erd\H{o}s-R{\'e}nyi (ER) model with edge probability $\frac\cc2$. Under $\Hc_1$, we keep $P$ as before whereas each entry of $Q$ is chosen independent and uniformly from $\{\frac\cc2-\gamma,\frac\cc2+\gamma\}$. 
    An appropriate choice of $\gamma\leq\frac\cc2$ leads to the result.
\end{proof}

The stated lower bound shows that at least $m\gtrsim \frac{n}{\ln n}$ samples are needed to test for separation in total variation distance, which is beyond the small sample regime that we are interested in.
In fact, in the case of constant $m$, one can improve the stated bound as $1-e^{-C''n} \leq \sc^* \leq 1$, where $C''$ is a constant that depends on $m$.

An impossibility result also holds in the case of symmetric KL-divergence, which has been effectively used for high dimensional discrete distributions, particularly Ising models~\citep{Daskalakis_2018_conf_SODA_ising}.
For any two probability mass functions $p$ and $q$, defined on the space of undirected $n$-vertex graphs, the symmetric KL-divergence is given by
\begin{align*}
    SKL(p,q) &= \sum_G p(G)\ln\left(\frac{p(G)}{q(G)}\right) + q(G)\ln\left(\frac{q(G)}{p(G)}\right),
\end{align*}
where the summation is over all unweighted undirected graphs on $n$ vertices.
Note that the above distance is unbounded even for finite $n$ since $SKL(p,q) = \infty$ if there exists $G_0$ such that $p(G_0)=0$ but $q(G_0)\neq0$.
We present the following result that demonstrates the impossibility of testing with respect to the symmetric KL-divergence when sample size $m$ is small.

\begin{prop}[$\sc^*$ for symmetric KL-divergence]
\label{prop_skl}
Let $\eta\in(0,1)$ and consider the problem in~\eqref{eqn_prob}--\eqref{eqn_prob_sets} with $d(P,Q) = SKL\big(\IER(P),\IER(Q)\big)$ and any $\cc\in(0,1)$. Then
\begin{displaymath}
\sc^* = \infty
\qquad\text{for~ } m\leq \displaystyle \frac{2}{\cc}\ln((1-\eta)n) \,.
\end{displaymath}
\end{prop}

\begin{proof}[Proof sketch]
    The basic technique for the proof is similar to Proposition~\ref{prop_tv}, but we choose $P$ and $Q$ such that $SKL\big(\IER(P),\IER(Q)\big)=\infty$, and yet the models are indistinguishable for small $m$.
    
    To be precise, under $\Hc_0$ we set $P=Q$ corresponding to an ER model.
    Under $\Hc_1$, we set $Q$ to be the same as $P$ except for a randomly chosen entry, for which $Q_{ij}=0\neq P_{ij}$. This implies that $\IER(P)$ and $\IER(Q)$ do not have a common support, which leads to $SKL\big(\IER(P),\IER(Q)\big)=\infty$.
    However, if the sample size is small (small $m$) or the graphs are sparse (small $\cc$), then the models cannot be distinguished.
\end{proof}

The above result shows that testing for separation in symmetric KL-divergence is impossible for $m\lesssim \ln n$ sample even when the graphs are  dense $(\cc\asymp1)$. However, the situation is worse in the case of sparse graphs. If $\cc\asymp \frac1n$, then at least $m\gtrsim n\ln n$ samples are necessary, which is worse than the condition for total variation distance.

Both Propositions~\ref{prop_tv} and~\ref{prop_skl} suggest that achieving a small sample complexity (small $m$) could be difficult under general difference alternatives, that is, if $d$ corresponds to distributional distances.
Hence, subsequent discussions focus only on matrix distances.
However, even in this case, the two-sample problem is not necessarily easily solvable for all distances or dissimilarities.
\begin{prop}[$\sc^*$ for zero norm / effect rarity]
\label{prop_zero}
Let $\eta\in(0,1)$ and consider the problem in~\eqref{eqn_prob}--\eqref{eqn_prob_sets} with $d(P,Q) = \Vert P-Q\Vert_0$ and any $\cc\in(0,1)$. Then
\begin{displaymath}
\sc^* = n(n-1)
\qquad\text{for all~ } m<\infty \,.
\end{displaymath}
\end{prop}

\begin{proof}[Proof sketch]
    The proof is straightforward since the entries $P$ and $Q$ can be arbitrarily close but still be unequal. Hence, the models may not be distinguishable though $\Vert P-Q\Vert_0 = n(n-1)$, which is the trivial upper bound.
\end{proof}

Proposition~\ref{prop_zero} may be viewed as a trivial extremity of the rare/weak effect studied in the context of multiple testing~\citep{Donoho_2015_jour_StatScience}. 
To put it simply, here we view the problem as testing $P_{ij}=Q_{ij}$ or $P_{ij}\neq Q_{ij}$ for every $i<j$, and the edge independence in $\IER$ graphs leads to a problem of multiple independent comparisons.
Proposition~\ref{prop_zero} states that if $\min\limits_{P_{ij}\neq Q_{ij}} |P_{ij} - Q_{ij}|$ is arbitrarily small, that is, the individual effects are arbitrarily weak, then they cannot be detected even when the effects are dense $\Vert P-Q\Vert_0 = n(n-1)$.
A more detailed analysis of the rare/weak effect in the sparse Bernoulli setting may be done by imposing a threshold $\min\limits_{P_{ij}\neq Q_{ij}} |P_{ij} - Q_{ij}|$ that characterises weakness of the effect. 
We do not discuss further on this effect, and instead, we proceed to other instances of~\eqref{eqn_prob}--\eqref{eqn_prob_sets}, where the problem can be solved for small $m$.

\section{Testing for separation in Frobenius norm}
\label{sec_frobenius}

The previous section focused on impossibility results, where $\sc^*$ is typically large (close to trivial upper bound) when $m$ is small.
In this section and the next one, we study two instances of the problem~\eqref{eqn_prob}--\eqref{eqn_prob_sets} where tests can be constructed even for small $m$.
Formally, we show that the lower bound for $\sc^*$ can be much smaller than the trivial upper bound, and subsequently, we propose two-sample tests to derive nearly matching upper bounds for $\sc^*$. 

We first quantify the separation in terms of Frobenius norm, that is, $d(P,Q) = \Vert P-Q\Vert_F$.
This is equivalent to viewing the adjacencies as $\binom{n}{2}$-dimensional Bernoulli vectors, and using two-sample test for high dimensional vectors --- a well-studied problem in the Gaussian case~\citep{Chen_2010_jour_AnnStat}.
We state the following bounds for the minimax separation $\sc^*$.
\begin{thm}[$\sc^*$ for Frobenius norm separation]
\label{thm_frotest}
Consider the two-sample problem~\eqref{eqn_prob}--\eqref{eqn_prob_sets}  with $d(P,Q) = \Vert P-Q\Vert_F$, any $\cc\in(0,1)$ and any $\eta\in(0,1)$.
There exist absolute constants $C_1,C_2\geq1$ such that:
\begin{enumerate}
\item $\displaystyle \frac{n\cc}{4} \leq \sc^* \leq n\cc\;$ for $m=1$,
\vskip1ex
\item $\displaystyle \bigg(\frac14\bigwedge\sqrt{\frac{\eta^2\ell_\eta}{8C_1}}\bigg)n\cc \leq \sc^* \leq n\cc\;$ for $m>1$ and $\cc \leq \displaystyle\frac{C_1}{\eta^2 m n}$, and
\vskip1ex
\item $\displaystyle \sqrt{\frac{\ell_\eta}{8} \frac{n\cc}{m}} \leq \sc^* \leq \sqrt{\frac{C_2}{\eta}\frac{n\cc}{m}}\;$ for $m>1$ and $\cc \geq \displaystyle\frac{C_1}{\eta^2 m n}\;$, 
\end{enumerate}
where $\ell_\eta = \sqrt{\ln\left(1+4(1-\eta)^2\right)}$.
Hence, assuming the allowable risk $\eta$ is fixed, we have $\sc^* \asymp n\cc$ for $m=1$ and $\sc^* \asymp n\cc \wedge \sqrt{\frac{n\cc}{m}}\,$ for $m\geq2$.
\end{thm}

Theorem~\ref{thm_frotest} provides a clear characterisation of the minimax separation $\sc^*$ (up to factors of $\eta$) when the distance between models is in terms of Frobenius norm.
The second and third statements deal with the case of $m>1$. 
In the ultra-sparse regime, that is $\cc \lesssim \frac{1}{mn}$, one observes a total of only $O(n)$ edges from the entire population of $2m$ graphs generated from either models. This information is insufficient for testing equality of models, and hence, it is not surprising that  $\sc^*\asymp n\cc$, which is the trivial upper bound. 
On the other hand, when $\cc \gtrsim \frac{1}{mn}$, we find a non-trivial separation rate indicating that the problem is solvable in this case.

The surprising finding of Theorem~\ref{thm_frotest} is that $\sc^*\asymp n\cc$ for $m=1$, which informally means that the problem is not solvable when one observes only $m=1$ sample from each model.
This result is significant since it shows that the problem of testing for separation in Frobenius norm is unsuitable in the setting of comparing between two large networks, for instance, the case of testing between two omics networks.

\subsection{Proof of Theorem~\ref{thm_frotest}}
We provide an outline of the proof of the above result highlighting the key technical lemmas.
We sketch their proofs here, and the detailed proofs can be found in the appendix.
To prove the lower bounds, we have the following result.

\begin{lem}[Necessary conditions for  detecting Frobenius norm separation]
\label{thm_frotest_LB}
For the testing problem~\eqref{eqn_prob}--\eqref{eqn_prob_sets} with $d(P,Q) = \Vert P-Q\Vert_F$ and for any $\eta\in(0,1)$, the minimax risk~\eqref{eqn_minimaxrisk}  is at least $\eta$ if either of the following conditions hold:
\begin{displaymath}
\textup{(i)~} \sc <  \frac{n\cc}{4} \bigwedge \sqrt{\frac{\ell_\eta}{8}\frac{n\cc}{m}}\,, \qquad \text{or} \qquad
\textup{(ii)~} m=1, ~\sc < {\frac{n\cc}{4}}\,.
\end{displaymath}
\end{lem}

\begin{proof}[Proof sketch]
    The proof follows the basic approach of Proposition~\ref{prop_tv}, and also uses same choice of $P,Q$.
    We set $P=Q$ corresponding to ER model with edge probability $\frac\cc2$ under $\Hc_0$. 
    For $\Hc_1$, we set the same $P$, but each entry of $Q$ is chosen independent and uniformly from $\{\frac\cc2-\gamma,\frac\cc2+\gamma\}$. 
    One can easily see that $Q$ is chosen uniformly from a set of $2^{n(n-1)/2}$ matrices, but for each choice of $Q$, we have $\Vert P-Q\Vert_F \approx n\gamma$.
    Hence, a choice of $\gamma \in \left(\frac{\sc}{n}, \frac\cc2\right]$ implies that the pair of $(P,Q)$ for every choice of $Q$ lies in $\Omega_1$.
    
    Subsequently, we use the techniques of \citet{Baraud_2002_jour_Bernoulli} to show that for $m\geq2$ and $\cc \gtrsim \frac{1}{mn}$, there is an appropriate choice $\gamma < \frac\cc2$ for which the random choice of $(P,Q)\in\Omega_1$ cannot be distinguished from the null case. 
    If $\cc \lesssim \frac{1}{mn}$, then the same situation occurs even for the choice $\gamma=\frac\cc2$.
    Finally for $m=1$, we observe that the same proof leads to the conclusion that the random choice of $(P,Q)\in\Omega_1$ is indistinguishable from the null case for any choice of $\gamma\leq \frac\cc2$. In particular, $\gamma=\frac\cc2$ leads to the claim in (ii).
    
    We elaborate on the distinction between the cases $m=1$ and $m>1$. 
    Consider the former case of $m=1$ under $\Hc_1$, where we have $G_1\sim \IER(P) = \text{ER}(\frac\cc2)$ and $H_1\sim \IER(Q)$ with $Q$ being chosen randomly as described above.
    Due the uniform choice of $Q$, one can easily verify that the probability of each edge in $H_1$ is $\frac\cc2$, which is same as that of $G_1$. 
    Hence, although the two graphs are sampled from different generative models with $\Vert P-Q\Vert_F > \sc$, they are essentially similar due to the random choice of $Q$.
    On the other hand, let $m=2$, $G_1,G_2\sim_{iid} \text{ER}(\frac\cc2)$ and $H_1,H_2\sim_{iid} \IER(Q)$ with $Q$ being random as before.
    Although $H_1,H_2$ are independent conditioned on the choice of $Q$, the two graphs are mutually dependent without the knowledge of $Q$ and hence, the population $\{H_1,H_2\}$ does not have the same distribution as $\{G_1,G_2\}$. 
\end{proof}

We continue with the proof of Theorem~\ref{thm_frotest}.
The lower bounds in the second and third statements of Theorem~\ref{thm_frotest} follow from condition (i) above by accounting for the conditions on $\cc$ and noting $C_1\geq1$ and $\ell_\eta \leq \sqrt{\ln 5}$.
For the upper bounds in first two statements, note that $\sc^*\leq n\cc$ trivially holds since $\Vert P-Q\Vert_F \leq n(\Vert P\Vert_{max} \vee \Vert Q\Vert_{max})$.
To derive the upper bound for the third case, we construct the following two-sample test.
Let $A_{G_1},\ldots,A_{G_m}$ and $A_{H_1},\ldots,A_{H_m}$ be the adjacency matrices of the $2m$ graphs.
We define
\begin{align}
\mhat &= \sum_{\substack{i,j=1 \\ i<j}}^n
\left(\sum\limits_{k\leq  m/2} (A_{G_k})_{ij} - (A_{H_k})_{ij}\right) \left(\sum\limits_{k> m/2} (A_{G_k})_{ij} - (A_{H_k})_{ij}\right),
\label{eqn_mhat} 
\\
\shat &= \sqrt{\sum_{\substack{i,j=1 \\ i<j}}^n
\left(\sum\limits_{k\leq  m/2} (A_{G_k})_{ij} + (A_{H_k})_{ij}\right) \left(\sum\limits_{k> m/2} (A_{G_k})_{ij} + (A_{H_k})_{ij}\right)},
\label{eqn_shat}
\end{align}
and consider the test
\begin{equation}
\Psi_F = \1\left\{ \frac{\mhat}{\shat}>  \frac{t_1}{\sqrt\eta} \right\} \cdot \1\left\{ \shat > \frac{t_2}{\eta^{3/2}} \right\} 
\label{eqn_fro_test}
\end{equation}
for some positive constants $t_1,t_2$, where $\1\{\cdot\}$ is the indicator function. We state the following guarantee for $\Psi_F$.

\begin{lem}[Sufficient conditions for detecting Frobenius norm separation]
\label{thm_frotest_UB}
Consider the testing problem with $d(P,Q) = \Vert P-Q\Vert_F$ and the test $\Psi_F$ in~\eqref{eqn_fro_test}. There exist absolute constants $t_1,t_2,C$ and $C'$ such that for any $\eta\in(0,1)$ and $\cc\in(0,1)$, if 
\begin{equation}
m\geq2 \qquad\text{and} \qquad \sc \geq \displaystyle \sqrt{\frac{C}{\eta}\frac{n\cc}{m}} ~\bigvee~ \frac{C'}{\eta^{3/2}m}, 
\label{eqn_fro_test_UB}
\end{equation}
then $R(\Psi_F,n,m,d,\sc,\cc) \leq \eta$.
\end{lem}

\begin{proof}[Proof sketch]
  The proof is based on concentration statements for $\widehat{\mu}$ and $\widehat{\sigma}$ derived via Chebyshev's inequality using
  $$\mu=\E[\mhat] = \frac{m^2}{8}\Vert P-Q\Vert_F^2,~~\sigma^2=\E[\shat^2] = \frac{m^2}{8}\Vert P+Q\Vert_F^2$$
  and bounds for $\Var[\mhat]$ and $\Var[\shat^2]$ in terms of $\Vert P-Q\Vert_F$ and $\Vert P+Q\Vert_F$.\\
  For the type-I-error, i.e. under $\Hc_0$, these concentration bounds lead to large enough choices for $t_1$ and $t_2$ such that at least one of the events
  $$\left\{ \frac{\mhat}{\shat}>  \frac{t_1}{\sqrt\eta} \right\}~~\mathrm{and}~~\left\{ \shat > \frac{t_2}{\eta^{3/2}} \right\}$$
  has small probability, which bounds the probability of the event $\{\Psi_F=1\}$.
  On the other hand, by construction, in order to control the type-II-error rate we want to ensure that both the events
  $$\left\{ \frac{\mhat}{\shat}\leq  \frac{t_1}{\sqrt\eta} \right\}~~\mathrm{and}~~\left\{ \shat \leq \frac{t_2}{\eta^{3/2}} \right\}$$
  have small probability under $\Hc_1$. Now, the requirement in \eqref{eqn_fro_test_UB} on $\rho$ guarantees that
  $$\mu\gtrsim\sigma\gtrsim\frac{1}{\sqrt{8}\eta^{3/2}},$$
  which allows us to rewrite these two events as large deviation statements of $\mhat$ and $\shat^2$ from their means as required.
\end{proof}
To conclude the proof of Theorem~\ref{thm_frotest}, observe that if $\cc\geq \frac{C'^2}{C\eta^2 m n}$, then the $\sqrt\frac{n\cc}{m}$ term in the above result dominates and we obtain the upper bound in the  third statement of Theorem~\ref{thm_frotest} with $C_1 = \frac{C'^2}{C}$.
Hence the theorem holds.

\subsection{Further remarks on Theorem~\ref{thm_frotest}}

As part of the proof of Theorem~\ref{thm_frotest}, we propose a test in~\eqref{eqn_fro_test} which provides the non-trivial upper bound in the third statement of Theorem~\ref{thm_frotest}.
Though this bound matches the corresponding lower bound up to factors of $\eta$, the difference is rather large with respect to $\eta$.
This is an artefact of the proof of Lemma~\ref{thm_frotest_UB} which is based on Chebyshev's inequality, and its effect can also be seen in the two thresholds $\frac{t_1}{\sqrt{\eta}}$ and $\frac{t_2}{\eta^{3/2}}$ defined in~\eqref{eqn_fro_test}, which can be very high for small $\eta$ limiting the practical usefulness of the test.
Below, we show that this can be improved using more refined concentration inequalities, but provides a sufficient condition that is weaker by a factor of $\ln n$.

\begin{prop}[Improving dependence on $\eta$]
\label{thm_frotest_UB2}
Consider the two sample problem of Theorem~\ref{thm_frotest} and assume $m\geq2$. Define the test
\begin{equation}
\Psi'_F = \1\left\{ \frac{\mhat}{\shat}>  t_1\ln\left(\frac{2}{\eta}\right) \ln\left(\frac{n}{\eta}\right) \right\} \cdot \1\left\{ \shat > t_2\ln^2\left(\frac{2}{\eta}\right)\ln\left(\frac{n}{\eta}\right) \right\} ,
\label{eqn_fro_test2}
\end{equation}
where $\mhat,\shat$ are as in~\eqref{eqn_mhat}--\eqref{eqn_shat}.
There exist constants $t_1,t_2,C$ and $C'$ such that for any $\eta\in(0,1)$ and $\cc\in(0,1)$, 
the test in~\eqref{eqn_fro_test2} has a risk at most $\eta$ whenever
\begin{equation}
\sc \geq \displaystyle C\ln\left(\frac{2}{\eta}\right)\sqrt{\frac{n\cc}{m}\ln\left(\frac{n}{\eta}\right)} ~\bigvee~ \frac{C'}{m}\ln^2\left(\frac{2}{\eta}\right)\ln\left(\frac{n}{\eta}\right).
\label{eqn_fro_test2_rho}
\end{equation}
\end{prop}

\begin{proof}[Proof Sketch]
    The proof is very close in spirit to that of Lemma \ref{thm_frotest_UB} above, but it is based on a much more involved concentration statement than Chebyshev's inequality (stated in the appendix). As a result, the test is based on the same test statistic $\frac{\mhat}{\shat}$ and differs from~\eqref{eqn_fro_test} only in the choice of thresholds.
    
    More specifically, due to the fact that $\mhat$ and $\shat^2$ are sums of products of sums with strong independence properties, we can derive concentration inequalities with logarithmic dependence on $\eta$ for them by repeated application of Bernstein's inequality. However, this comes with the additional $\ln n$ factor which can be seen in equation \eqref{eqn_fro_test2_rho} above.
\end{proof}

A key feature of both tests is that they are adaptive, that is, they do not require specification of the sparsity parameter $\cc$.
We highlight the importance of this property in the following remark.

\begin{rem}[Adaptivity of proposed tests]
\label{rem_adaptive}
The testing problem in~\eqref{eqn_prob}--\eqref{eqn_prob_sets} is defined with respect to the sparsity parameter $\cc$, which in turn governs the minimax separation rate $\sc^*$.
It is not hard to convince one that for any $P,Q$, it is impossible to estimate $\Vert P\Vert_{max}\vee \Vert Q\Vert_{max}$ from few observations (small $m$), and setting $\cc=1$ is clearly sub-optimal for sparse graph models.
Hence, it is desirable to construct tests that do not require knowledge of $\cc$, and both tests in~\eqref{eqn_fro_test} and~\eqref{eqn_fro_test2} are adaptive in this sense.
Adaptivity of these tests are achieved by estimating $\Vert P+Q\Vert_F$, which is a lower bound for $2n\cc$.
\end{rem}

\section{Testing for separation in operator norm}
\label{sec_operator}

In this section, we study the two sample testing problem where $d(P,Q) = \Vert P-Q\Vert_{op}$, and provide bounds on the minimax separation $\sc^*$ for all $m\geq1$.
An interesting finding of this section is that one can obtain a non-trivial minimax separation rate even for $m=1$, that is, the problem is indeed solvable even with a single observation from each model.
Our main result is the following.

\begin{thm}[$\sc^*$ for operator norm separation]
\label{thm_optest}
Consider the two-sample problem~\eqref{eqn_prob}--\eqref{eqn_prob_sets}  with $d(P,Q) = \Vert P-Q\Vert_{op}$, and any $m\geq1$, $\cc\in(0,1)$. Let $\eta\in(0,1)$ and $\ell'_\eta = \frac{\ell_\eta}{\sqrt{8}} \wedge \frac{1}{16}$.
There exist constants $C,C'\geq1$ such that:
\begin{enumerate}
\item $\displaystyle \frac{n\cc}{4} \leq \sc^* \leq n\cc\;$ for $\cc \leq \displaystyle\frac{\ell'^2_\eta}{16 m n}\;$, and
\vskip1ex
\item $\displaystyle \ell'_\eta\sqrt{\frac{n\cc}{m}} \leq \sc^* \leq \sqrt{\frac{n\cc}{m}}\left(C\sqrt{\ln\left(\frac{n}{\eta}\right)} ~\bigvee~ \frac{4C'}{\ell'_\eta}\ln\left(\frac{n}{\eta}\right)\right) \;$ otherwise. 
\end{enumerate}
\end{thm}

The theorem shows that the problem is not solvable in the ultra-sparse regime, that is, $\cc \lesssim \frac{1}{mn}$.
However, beyond this regime there is a non-trivial separation rate, which Theorem~\ref{thm_optest} finds up to a factor of $\ln n$.
It is natural to ask whether the additional logarithmic factor is necessary. 
Later in this section, we refine Theorem~\ref{thm_optest} to remove the $\ln n$ term in the upper bound (see Corollary~\ref{cor_optest_improved}).
This is achieved by using a non-adaptive test, which has prior knowledge of $\cc$ (see Proportion~\ref{thm_optest_UB2}).

\subsection{Proof of Theorem~\ref{thm_optest}}
The lower bounds in the theorem are due to the following necessary condition.

\begin{lem}[Necessary condition for  detecting operator norm separation]
\label{thm_optest_LB}
For the testing problem~\eqref{eqn_prob}--\eqref{eqn_prob_sets} with $d(P,Q) = \Vert P-Q\Vert_{op}$, $\cc\in(0,1)$ and $m\geq1$,  and for any $\eta\in(0,1)$, the minimax risk~\eqref{eqn_minimaxrisk}  is at least $\eta$ if 
\begin{displaymath}
\sc < \frac{n\cc}{4} \bigwedge \ell'_\eta \sqrt{\frac{n\cc}{m}} \;.
\end{displaymath}
\end{lem}

\begin{proof}[Proof sketch]
    The proof follows the technique of~\citet{Baraud_2002_jour_Bernoulli} to derive minimax lower bounds as used in the previous results in this paper.
    Hence, we mainly focus on the choice of $P,Q$ used in the present proof.
    We set $P=Q$ corresponding to ER model with edge probability $\frac\cc2$ under $\Hc_0$. 
    For $\Hc_1$, we set the same $P$, but $Q$ is randomly chosen in the following way. 
    We partition the vertices randomly into two groups, and set $Q_{ij} = \frac\cc2 + \gamma$ for $i,j$ belonging to the same group, and $Q_{ij} = \frac\cc2 - \gamma$ otherwise.
    The random choice of $Q$ is due to randomly sampling one of $2^{n-1}$ possible splits of the vertex set.
    
    One can see that for each choice of $Q$, we have $\Vert P-Q\Vert_{op} = \gamma(n-1)$.
    Hence, a choice of $\gamma \in \left(\frac{\sc}{n-1}, \frac\cc2\right]$ implies that the pair of $(P,Q)$ for every choice of $Q$ lies in $\Omega_1$.
    Now similar to the proof of Lemma~\ref{thm_frotest_LB}, we show for any $m\geq1$ and $\cc \gtrsim \frac{1}{mn}$, there is an appropriate choice $\gamma < \frac\cc2$ for which the random choice of $(P,Q)\in\Omega_1$ cannot be distinguished from the null case. 
    If $\cc \lesssim \frac{1}{mn}$, then the same situation occurs even for the choice $\gamma=\frac\cc2$, which leads to the claim.
\end{proof}

We now prove the upper bounds in Theorem~\ref{thm_optest}.
For this, we note that $\Vert P-Q\Vert_{op} \leq \Vert P-Q\Vert_{row} \leq n\cc$ is the trivial upper bound for $\sc^*$. 
To prove the non-trivial upper bound, we consider the following test:
\begin{equation}
\Psi_{op} = \1\left\{ \frac{\Vert S^- \Vert_{op}}{\sqrt{\Vert S^+ \Vert_{row}}} > t_1 \sqrt{\ln\left(\frac{n}{\eta}\right)} \right\}\cdot \1\left\{ \Vert S^+\Vert_{row} > t_2\ln\left(\frac{n}{\eta}\right)\right\} ,
\label{eqn_op_test}
\end{equation}
where
\begin{equation}
S^- = \sum_{k=1}^m A_{G_k} - A_{H_k} \qquad \text{and} \qquad S^+ = \sum_{k=1}^m A_{G_k} + A_{H_k} \,.
\label{eqn_Spm}
\end{equation}

We have the following sufficient condition for the two-sample test $\Psi_{op}$.
\begin{lem}[Sufficient condition for detecting operator norm separation]
\label{thm_optest_UB}
Consider the testing problem with $d(P,Q) = \Vert P-Q\Vert_{op}$ and the test $\Psi_{op}$ in~\eqref{eqn_op_test}. There exist absolute constants $t_1,t_2,C$ and $C'$ such that for any $\eta\in(0,1)$, $\cc\in(0,1)$ and $m\geq1$, the risk $R(\Psi_{op},n,m,d,\sc,\cc) \leq \eta$ if 
\begin{displaymath}
\sc \geq \displaystyle C\sqrt{\frac{n\cc}{m}\ln\left(\frac{n}{\eta}\right)} ~\bigvee~ \frac{C'}{m}\ln\left(\frac{n}{\eta}\right) \,.
\end{displaymath}
\end{lem}

\begin{proof}[Proof sketch]
    The proof mainly controls the type-I and type-II error rates for $\Psi_{op}$. 
    We achieve this through the matrix Bernstein inequality~\citep{Tropp_2012_jour_FOCM,Oliveira_2009_arxiv_0600}, which in the present case guarantees that
    \begin{align*}
        \Vert S^- - m(P-Q)\Vert_{op} \lesssim \sqrt{m \Vert P+Q\Vert_{row} \ln n}
    \end{align*}
    with high probability if $\Vert P +Q\Vert_{row} \gtrsim \frac{\ln n}{m}$. 
    We also use a Bernstein-type concentration result to ensure that $\Vert S^+\Vert_{row} \asymp m\Vert P +Q\Vert_{row}$ for $\Vert P +Q\Vert_{row}\gtrsim \frac{\ln n}{m}$, and $\Vert S^+\Vert_{row} \lesssim \ln n$ otherwise. 
    
    We bound the type-I error by noting that under $\Hc_0$, that is $P=Q$, the above concentration results imply that with high probability:
    \begin{itemize}
        \item $\Vert S^-\Vert_{op} \lesssim \sqrt{\Vert S^+\Vert_{row} \ln n}$ for $\Vert P +Q\Vert_{row} \gtrsim \frac{\ln n}{m}$, and
        \item $\Vert S^+\Vert_{row} \lesssim \ln n$ for $\Vert P +Q\Vert_{row} \lesssim \frac{\ln n}{m}$.
    \end{itemize}
    Hence, we have $\Psi_{op}=0$ with high probability for a suitable choices of $t_1, t_2$. 
    
    On the other hand, we control the type-II error in the following way.
    Under $\Hc_1$, we have $\Vert P-Q\Vert_{op} > \sc$ with $\sc \geq  \frac{C'}{m}\ln \left(\frac{n}{\eta}\right)$,  which implies that 
    \begin{align*}
        \Vert P +Q\Vert_{row} \geq \Vert P-Q\Vert_{op} \gtrsim \frac{\ln n}{m}\,.
    \end{align*}
    So the second indicator in $\Psi_{op}$ is guaranteed to be 1 for $C'$ large enough.
    This condition also allows us to use the matrix Bernstein inequality which, along with the reverse triangle inequality, gives
     \begin{align*}
        \Vert S^-\Vert_{op}
        &\geq m\Vert P-Q\Vert_{op} - \Vert S^- - m(P-Q)\Vert_{op}
        \\&\gtrsim m\sc - \sqrt{m\Vert P+Q\Vert_{row}\ln n}.
    \end{align*}
    We use the stated assumption $\sc \geq C\sqrt{\frac{n\cc}{m}\ln\left(\frac{n}{\eta}\right)}$ and the fact $\Vert S^+ \Vert_{row} \asymp m\Vert P+Q\Vert_{row}$ to prove that the first indicator in $\Psi_{op}$ is also true with high probability for large enough $C$. This bounds the type-II error rate.
\end{proof}

To conclude the proof of Theorem~\ref{thm_optest}, we note that the upper bound in the second statement is obtained by observing that $\frac{C'}{m} \leq \frac{4C'}{\ell'_\eta}\sqrt{\frac{n\cc}{m}}$ for $\cc \geq \frac{\ell'^2_\eta}{16mn}$.
We note that the above analysis of the test $\Psi_{op}$ is based on the matrix Bernstein inequality~\citep{Tropp_2012_jour_FOCM}, which makes the upper bound worse factor of $\ln n$.
One may alternatively use other concentration results based on refinements of the trace method~\citep{Lu_2013_jour_EJComb,Bandeira_2016_jour_AnnProb} to obtain slightly different sufficient conditions in Lemma~\ref{thm_optest_UB}.
For instance, an use of \citet[][Corollary 3.9]{Bandeira_2016_jour_AnnProb} does provide a rate of $\sqrt{\frac{n\cc}{m}}$, but such bounds hold only if $\cc\gtrsim \frac{\textup{polylog}(n)}{n}$.

\subsection{An optimal non-adaptive test}
One of the objectives of this work is to determine whether it is possible to test between two large graphs, that is, the case of $m=1$.
Theorem~\ref{thm_optest} provides an affirmative answer to this question, but our proposed test $\Psi_{op}$~\eqref{eqn_op_test} is not optimal since the sufficient condition on $\sc$ is worse than the necessary condition by a factor of $\ln n$.
As we note above, this is a consequence of our use of matrix Bernstein inequality in the proof of Lemma~\ref{thm_optest_UB}, and leads to the question whether one can improve the result by using more sharp concentration techniques known in the case of random graphs.
We show that this is indeed true, at least for $m=1$, and can be shown using concentration of trimmed or regularised adjacency matrices~\citep{Le_2015_arxiv}. 

Assume $m=1$, and let $G\sim\IER(P)$ and $H\sim\IER(Q)$ be the two random graphs. 
Also assume that $\cc$ is specified.
For some constant $c\geq6$, let $A'_G$ be the adjacency matrix of the graph obtained by deleting all edges in $G$ that are incident on vertices with degree larger than $cn\cc\ln(\frac2\eta)$. Similarly, we obtain $A'_H$ from $H$.
Define a test as
\begin{equation}
\Psi'_{op} = \1\left\{ \Vert A'_G - A'_H \Vert_{op} > t\sqrt{n\cc}\ln^2\left(\frac2\eta\right)\right\}
\label{eqn_op_test2}
\end{equation}
for some constant $t>0$. We have the following guarantee for $\Psi'_{op}$.

\begin{prop}[Optimality for $m=1$]
\label{thm_optest_UB2}
Consider the testing problem with $d(P,Q) = \Vert P-Q\Vert_{op}$, $m=1$ and $\cc > \frac{10}{n}$.
There exist absolute constants $c,t$ such that for any $\eta\in(0,1)$, the maximum risk of $\Psi'_{op}$ is at most $\eta$ if 
\begin{displaymath}
\sc \geq \displaystyle 2t\sqrt{n\cc}\ln^2\left(\frac{2}{\eta}\right) \,.
\end{displaymath}
Hence, assuming $\eta$ is fixed, $\sc^*\asymp \sqrt{n\cc}$ for $\cc> \frac{10}{n}$ whereas $\sc^* \asymp n\cc$ for $\cc\leq\frac{10}{n}$.
\end{prop}

\begin{proof}[Proof sketch]
    The proof differs from that of Lemma~\ref{thm_optest_UB} in the use of a different matrix concentration inequality.
    We rely on the concentration of the trimmed adjacency matrix~\citep{Le_2015_arxiv}. 
    In the present context, the result states that for $\cc>\frac{10}{n}$ and $G\sim\IER(P)$ with $\Vert P\Vert_{max} \leq \cc$, 
    \begin{align*}
        \Vert A'_G - P\Vert_{op} \leq C\sqrt{n\cc}\ln^2\left(\frac{2}{\eta}\right)
    \end{align*}
    with probability $1-\frac{\eta}{4}$ for some absolute constant $C>0$. 
    
    For $P=Q$, we have with probability $1-\frac\eta2$,  
    \begin{align*}
        \Vert A'_G - A'_H \Vert_{op} \leq \Vert A'_G - P\Vert_{op} + \Vert A'_H - Q\Vert_{op} \leq 2C\sqrt{n\cc}\ln^2\left(\frac{2}{\eta}\right),
    \end{align*}
    and hence, setting $t=2C$ ensures that the type-I error rate is smaller than $\frac\eta2$.
    Similarly, it follows that under $\Hc_1$,
    \begin{align*}
        \Vert A'_G - A'_H \Vert_{op} &\geq \Vert P-Q\Vert_{op} - \Vert A'_G - P\Vert_{op} + \Vert A'_H - Q\Vert_{op} 
        \\&\geq \sc - 2C\sqrt{n\cc}\ln^2\left(\frac{2}{\eta}\right),
    \end{align*}
    with probability $1-\frac\eta2$.
    Hence, $\Vert A'_G - A'_H \Vert_{op}$ exceeds the threshold for $\sc>4C\sqrt{n\cc}\ln^2\left(\frac{2}{\eta}\right)$, and we have a similar bound on the type-II error rate.
\end{proof}

While $\Psi'_{op}$ indeed achieves optimality, it comes at the cost of prior knowledge of $\cc$. 
As we note in Remark~\ref{rem_adaptive}, this is an unreasonable restriction and hence, the sub-optimal test $\Psi_{op}$~\eqref{eqn_op_test} may be preferable due to its adaptivity.
We do not know if an adaptive optimal operator norm based test can be constructed since existing concentration bounds for sparse IER graphs are typically in terms of the largest edge probability, which is hard to estimate.

\subsection{Improving the bounds in Theorem~\ref{thm_optest}}

As discussed above, the non-adaptive test $\Psi'_{op}$ helps in improving the upper bound on $\sc^*$ for $m=1$.
We have not studied whether $\Psi'_{op}$~\eqref{eqn_op_test2} and Proposition~\ref{thm_optest_UB2} extend to the case of $m>1$.
However, the following corollary shows that the results of Sections~\ref{sec_frobenius} and~\ref{sec_operator} can be combined to remove the undesirable $\ln n$ factor in the upper bound of $\sc^*$ in Theorem~\ref{thm_optest} for any $m\geq1$.
For convenience, let the allowable risk $\eta$ be a constant, which we ignore in the statement of the result.

\begin{cor}[Improved bounds on $\sc^*$ for operator norm separation]
\label{cor_optest_improved}
Consider the two-sample problem~\eqref{eqn_prob}--\eqref{eqn_prob_sets}  with $d(P,Q) = \Vert P-Q\Vert_{op}$, and any $m\geq1$, $\cc\in(0,1)$. 
There exists a constant $C>0$ such that:
\begin{center}
$\displaystyle \sc^* \asymp n\cc\;$ for $\cc \leq \displaystyle\frac{C}{m n} \;,\qquad$ and
$\qquad \displaystyle \sc^* \asymp \sqrt{\frac{n\cc}{m}}\;$ otherwise.
\end{center}
\end{cor}

\begin{proof}[Proof]
    The lower bounds on $\sc^*$ follow from Theorem~\ref{thm_optest}, and hence it remains to prove that
    \begin{displaymath}
     \sc^* \lesssim n\cc \bigwedge \sqrt{\frac{n\cc}{m}} \qquad \text{for all } m\geq1.
    \end{displaymath}
    We claim that this is a direct consequence of the upper bounds in Proposition~\ref{thm_optest_UB2} and Theorem~\ref{thm_frotest}.
    To see this, consider the following test:
    \begin{center}
        $\Psi = \Psi'_{op}$ for $m=1$, $\qquad$ and  $\qquad \Psi = \Psi_F$ for $m\geq2$.
    \end{center}
    For $m=1$, Proposition~\ref{thm_optest_UB2} leads to the above stated upper bound.
    
    For $m\geq2$, note that given $\cc$ and $\sc$, the two-sample problems under Frobenius and operator norm separation have the same null set $\Omega_0$~\eqref{eqn_prob_sets},
    whereas the sets under alternative hypotheses are
    \begin{align*}
        \Omega_1^{op}&= \{(P,Q): \Vert P-Q\Vert_{op}>\sc, \Vert P\Vert_{max} \vee \Vert Q\Vert_{max} \leq \cc\}, \text{ and}
        \\
        \Omega_1^{F} &= \{(P,Q): \Vert P-Q\Vert_F>\sc, \Vert P\Vert_{max} \vee \Vert Q\Vert_{max} \leq \cc\}.
    \end{align*} 
    We use the relation $\Vert P-Q\Vert_F \geq \Vert P-Q\Vert_{op}$ to observe that $\Omega_1^{op} \subset \Omega_1^{F}$. As a consequence, the maximum risk of $\Psi_F$  under the two different distances can be related as
    \begin{align*}
        R(\Psi_F,n,m,\Vert\cdot\Vert_{op},\sc,\cc) \leq R(\Psi_F,n,m,\Vert\cdot\Vert_{F},\sc,\cc),
    \end{align*}
    that is, if $\sc$ is fixed and $\Psi_F$ is used for the problem of testing difference in operator norm, then it achieves a smaller worst-case risk than what it achieves for the the problem of testing difference in Frobenius norm.
    This implies that the minimax separation for operator norm is at most the minimax separation for Frobenius norm. 
    The claim now follows from the upper bounds in Theorem~\ref{thm_frotest} for $m\geq2$.
\end{proof}

\section{Discussion}
\label{sec_discussions}

In this section, we discuss related graph two-sample problems, and comment on practical two-sample tests.

\subsection{A note on the sparsity condition}

The notion of sparsity has no formal definition in the context of random graphs, unlike sparsity in the signal detection literature.
The informal definition of a sparse graph is a graph where the number of edges are not arbitrarily large compared to the number of vertices $n$.
The sparsity condition used in~\eqref{eqn_prob_sets}, that is $\Vert P\Vert_{max} \leq \cc$, is one approach to define sparse graphs. More precisely, this condition implies that the expected number of edges is at most $n^2\cc$, and in particular, for $\cc\asymp \frac1n$, we obtain graphs where the number of edges grows linearly with $n$.
However, one can induce sparsity through alternative restrictions on $P$:
\begin{align*}
&\text{(i)}~~\sum_{ij} P_{ij} \leq \cc_1  &&\text{(iii)}~~ \Vert P \Vert_{row} \leq \cc_3
\\&\text{(ii)}~~ \Vert P\Vert_F \leq \cc_2  &&\text{(iv)}~~ \Vert P \Vert_{0} \leq \cc_4
\end{align*}
among others. 
We note the $\cc_i$'s are of different order than $\cc$ used in~\eqref{eqn_prob_sets}.
The first condition bounds the expected number of edges, while condition (iii) provides graphs with bounded expected degrees.
The last restriction is along the lines of the signal detection literature since it implies that at most $\cc_4$ entries in $P$ are non-zero, and results in random graphs with absolutely bounded number of edges (not only in expectation).

It is natural to ask to whether our results also extend to the case where sparsity is controlled through any one of the conditions.
We do not provide a complete characterisation in each case, but present two corollaries of Theorems~\ref{thm_frotest} and~\ref{thm_optest} which show that some of our results easily extend to alternative notions of sparsity.

\begin{cor}[$\sc^*$ under Frobenius norm sparsity]
Consider the problem of testing between the following two hypotheses sets
\begin{align*}
\Omega_0 &= \{ P=Q, ~\Vert P\Vert_{F} \leq \cc_2\}, \qquad\textup{and}
\\
\Omega_1 &= \{ \Vert P-Q\Vert_F > \sc, ~\Vert P\Vert_{F} \vee \Vert Q\Vert_{F} \leq \cc_2\}. 
\end{align*}
The bounds on minimax separation stated in Theorem~\ref{thm_frotest} hold in this case with the substitution $\cc_2 = n\cc$. 
\end{cor}

\begin{proof}[Proof sketch]
The proof of the corollary follows immediately from that of Theorem~\ref{thm_frotest}. 
The substitution $\cc_2=n\cc$ is a consequence of the relation $\Vert P\Vert_F \leq n \Vert P\Vert_{max}$.
\end{proof}

We also have a result in the case of graphs with bounded expected degrees, which can be similarly derived from Theorem~\ref{thm_optest}.
We do not know if the more precise rates given in Corollary~\ref{cor_optest_improved} can be extended to this setting.

\begin{cor}[$\sc^*$ under row sum norm sparsity]
Consider the problem of testing between the following two hypotheses sets
\begin{align*}
\Omega_0 &= \{ P=Q, ~\Vert P\Vert_{row} \leq \cc_3\}, \qquad\textup{and}
\\
\Omega_1 &= \{ \Vert P-Q\Vert_{op} > \sc, ~\Vert P\Vert_{row} \vee \Vert Q\Vert_{row} \leq \cc_3\}. 
\end{align*}
The bounds on minimax separation stated in Theorem~\ref{thm_optest} hold in this case with the substitution $\cc_3 = n\cc$. 
\end{cor}

\subsection{On the practicality of proposed tests}
The theme of this paper has been to explore different separation criteria for which the graph two sample testing problem can be solved for a small population size.
In the process of addressing this question, we suggest adaptive tests $\Psi_F$~\eqref{eqn_fro_test}, $\Psi'_F$~\eqref{eqn_fro_test2} and $\Psi_{op}$~\eqref{eqn_op_test} that also turn out to be near-optimal for the problems of detecting separation in Frobenius or operator norms.
However, the practical applicability of these tests have not been discussed so far.

We note that in practice, one is more interested in the testing problem $P=Q$ vs. $P\neq Q$, and hence, a more basic question that needs to be addressed is --- Which separation criterion should be used?
The findings of this paper suggest that for $m=1$, operator norm separation is a possible choice, whereas other distances like total variation distance and Frobenius norm should not be considered.
For $m>1$ but small, we show that one could detect separation in Frobenius norm using $\Psi_F$~\eqref{eqn_fro_test} or detect separation in operator norm using $\Psi_{op}$~\eqref{eqn_op_test}.
We compare the relative merits of both tests in terms of sample complexity in the following way.

\begin{rem}[Comparison between $\Psi_F$ and $\Psi_{op}$]
\label{rem_compare}
Consider $P$ and $Q$ such that $P\neq Q$ and $\Vert P\Vert_{max}\vee \Vert Q\Vert_{max} \leq \cc$.
Ignoring constants and terms involving $\eta$, Lemma~\ref{thm_frotest_LB} shows that  $m\gtrsim \frac{n\cc}{\Vert P-Q\Vert_F^2}$ samples are necessary to distinguish between the two models.
On the other hand, Lemma~\ref{thm_optest_LB} shows that $m\gtrsim \frac{n\cc}{\Vert P-Q\Vert_{op}^2}$ samples are needed to distinguish between $P,Q$.
Since $\Vert P-Q\Vert_F \geq \Vert P-Q\Vert_{op}$, testing for Frobenius norm separation is easier than testing for separation in operator norm for $m>1$.
In other words, one may expect $\Psi_F$ to have a smaller risk than $\Psi_{op}$.
\end{rem}

However, the tests $\Psi_F$, $\Psi'_F$ and $\Psi_{op}$ require $n$ to be very large or the graphs to be dense to  achieve a small risk, and hence have limited applicability in moderate-sized problems. It is known that the practical applicability of concentration based tests can be improved by using bootstrapped variants~\citep{Gretton_2012_jour_JMLR}, which approximate the null distribution by generating bootstrap samples through random mixing of the two populations. 
Simulations, not included in this paper, show that permutation based bootstrapping provides a reasonable rejection rate for moderate sample size $(m\geq10)$, but such bootstrapping is not effective for smaller $m$, for instance $m=2$.
Furthermore, the relative merit $\Psi_F$ (or $\Psi'_F$) over $\Psi_{op}$, as suggested by Remark~\ref{rem_compare}, could not be verified in case of the bootstrapped variants.

When $m = 1$ and the population adjacency is of low rank, \citet{Tang_2016_jour_JCompGraphStat} suggest an alternative bootstrapping principle based on estimation of the population adjacency $P$ and then drawing bootstrap samples from estimate of $P$. 
Simulations in \citet{Ghoshdastidar_2018_conf_neurips} show that this procedure works to some extent for dense IER graphs but only when $P$ has a small (known) rank. When the rank is unknown or, more generally, if $P$ does not have a low rank, such bootstrapped tests are not necessarily reliable.
 
In the related work~\citep{Ghoshdastidar_2018_conf_neurips}, we explore alternative possibilities for constructing practical tests derived from Frobenius norm or operator norm based test statistics that work even for $m=1,2$.
These tests use statistics similar to the ones studied in the present work, but are based on asymptotic null distributions that hold approximately or under  stronger assumptions.
For instance, \citet{Ghoshdastidar_2018_conf_neurips} show that under $\Hc_0$ and assumptions on the edge density of the graphs, the Frobenius norm based statistic $\frac{\mhat}{\shat}$ (see~\eqref{eqn_mhat}--\eqref{eqn_shat}) is asymptotically dominated by a standard normal random variable as $n\to\infty$.
Based on this, we propose an asymptotic distribution based test that is powerful for all $m\geq2$ and moderately sparse graphs, and is reliable even in the case of real networks. 
For $m=1$, we consider the operator norm of re-scaled version of $A_{G_1}-A_{H_1}$, which approximately follows the Tracy-Widom law under $\Hc_0$ as $n\to\infty$.
The practical applicability of these tests stem from the fact that they do not explicitly rely on concentration inequalities that lead to large thresholds, as in the present paper, nor do they use bootstrapping strategies, which often require large sample sizes or assumptions on the graph model.
Thus, our related work provides more practically useful tests whose statistical guarantees hold either approximately or under additional assumptions.\footnote{The implementations for the practical tests proposed in~\citet{Ghoshdastidar_2018_conf_neurips} as well as bootstrapped variants of tests studied in the present paper are available at: \\ \url{https://github.com/gdebarghya/Network-TwoSampleTesting}}

\subsection{Minimax separation under structural assumptions}

The present paper studies the two-sample problem for $\IER$ graphs, where the population adjacency matrices do not have any structural restriction.
In other words, we study a hypothesis testing problem in a dimension of $\binom{n}{2}$ with a sample size of $m$. 
Under this broad framework, Theorem~\ref{thm_frotest} shows that the minimax separation in Frobenius norm $\sc^*_F$ is given by
\begin{displaymath}
 \sc^*_F \asymp \left\{ \begin{array}{ll}
      n\cc & \text{for } m=1, \text{ and}  \\
      \displaystyle n\cc\bigwedge\sqrt{\frac{n\cc}{m}} & \text{for } m>1.
 \end{array}\right.
\end{displaymath}
Similarly, Corollary~\ref{cor_optest_improved} shows that the minimax separation in operator norm
\begin{displaymath}
 \sc^*_{op} \asymp n\cc \bigwedge \sqrt{\frac{n\cc}{m}} \quad \text{ for all } m\geq1.
\end{displaymath}

It is natural to ask if these rates decrease if we further impose structural assumptions on the population adjacencies thereby effectively reducing the problem dimension. 
In this section, we provide some initial results in this direction.
We impose the structural assumptions by restricting the possible values for the population adjacency matrices.
Formally, we define $\widetilde{\mathbb{M}}_n \subset \Mb$ as the set of symmetric matrices in $[0,1]^{n\times n}$, whose diagonal entries are zero and satisfy additional structural assumptions (specified below).
Subsequently, the graph two-sample problem, restricted to a special graph class, can be stated as the problem of testing between 
\begin{displaymath}
 \Hc_0: (P,Q) \in \Omega_0 \cap \left(\widetilde{\mathbb{M}}_n \times \widetilde{\mathbb{M}}_n\right)   
 \quad \text{vs} \quad
 \Hc_1: (P,Q) \in \Omega_1 \cap \left(\widetilde{\mathbb{M}}_n \times \widetilde{\mathbb{M}}_n\right),
\end{displaymath}
where $\Omega_0$ and $\Omega_1$ are as defined in the original problem~\eqref{eqn_prob}--\eqref{eqn_prob_sets}. 

We begin with the most simple case of Erd\H{o}s-R{\'e}nyi (ER) graphs, where the restricted set $\widetilde{\mathbb{M}}_n$ corresponds to the symmetric matrices  whose diagonal entries are zero and all off-diagonal entries are identical.
Note that each distribution is modelled by a single parameter, and hence, we have a hypothesis testing problem in one dimension.
We present following result on minimax separation for Frobenius and operator norms in this setting.
We simplify the statement by ignoring absolute constants including the allowable risk $\eta$, which is assumed to be fixed.

\begin{prop}[Minimax separation for testing ER graphs]
\label{prop_er}
Consider the graph two-sample problem restricted to ER graphs with specified $n,m$ and $\cc$.
The minimax separation rates for Frobenius norm $\sc^*_F$ and for operator norm $\sc^*_{op}$ satisfy
\begin{displaymath}
 \sc^*_F \asymp \sc^*_{op} \asymp n\cc \bigwedge \sqrt{\frac{\cc}{m}} \qquad \text{for all } m\geq1. 
\end{displaymath}
In particular, in the sparsity regime $\cc \gtrsim \frac{1}{n^2m}$, the minimax separation is much smaller than the corresponding rates for $IER$ graphs.
\end{prop}

\begin{proof}[Proof sketch]
    The proof is relatively simple, and borrows ideas from the Theorem~\ref{thm_frotest}. Hence, we only provide a brief sketch of the proof.
    
    We consider the two-sample problem with $G_1,\ldots,G_m \sim_{iid} \text{ER}(p)$ and $H_1,\ldots,H_m \sim_{iid} \text{ER}(q)$, where $p,q\leq \cc$ denote the edge probabilities in either models.
    Note that the result can be equivalently stated as: the minimax separation between $p$ and $q$ is $\cc \wedge \frac1n\sqrt{\frac{\cc}{m}}$.
    
    To derive a lower bound, we use the approach stated in the previous results and define $p,q$ as follows. Under $\Hc_0$, we set $p=q=\frac\cc2$ whereas under $\Hc_1$, we let $p=\frac\cc2$ and $q$ is randomly selected from $\{p-\gamma,p+\gamma\}$ for some $\gamma\leq \frac\cc2$. 
    It turns out that $\gamma \asymp \frac\cc2 \wedge \frac1n\sqrt{\frac{\cc}{m}}$ is an appropriate choice, which corresponds to the claimed lower bound for minimax separation.
    
    The upper bound is obtained by using a simple test that compares the edge densities estimated from the two population.
    Define a test
    \begin{displaymath}
     \Psi = \1\big\{ \left| \widehat{p} - \widehat{q}\right| > t_1 \big\} \cdot \1\big\{ \left(\widehat{p} + \widehat{q}\right) > t_2 \big\},
    \end{displaymath}
    where  $\widehat{p} = \displaystyle \frac{1}{m\binom{n}{2}} \sum_k \sum_{i<j} (A_{G_k})_{ij}$ and $\widehat{q} = \displaystyle \frac{1}{m\binom{n}{2}} \sum_k \sum_{i<j} (A_{H_k})_{ij}$ are the two edge density estimates, and $t_1,t_2$ are suitably defined thresholds.
    The proof strategy of Lemma~\ref{thm_frotest_UB} combined with judicious choice of thresholds provides the desired upper bound claim in the result. 
\end{proof}

We next increase the complexity of the problem by considering stochastic block model with 2 classes (2-SBM).
This class of graphs has been studied in the context of one-sample hypothesis testing, particularly for detecting community structure in graphs and estimating the number of communities in a block model~\citep{Bickel_2016_jour_JRStatistSocB,Lei_2016_jour_AnnStat,Gao_2017_arxiv_06742}.
We consider typical 2-SBM graphs, which are characterized by a binary vector denoting the communities of the nodes as well as the \textit{within community} and \textit{inter community} edge probabilities.
Formally, this is represented by the set $\widetilde{\mathbb{M}}_n \subset \Mb$ of matrices which can be transformed to have a 2 classes block structure through row/column permutations.
Furthermore, the off-diagonal entries in each matrix can take at most two distinct values. 
It is easy to observe that each matrix is governed by $n+2$ parameters, and the problem has a dimension much smaller than $\binom{n}{2}$-dimensional $\IER$ problem.
However, the following result shows that the minimax separation does not decrease in this case compared to the general $\IER$ setting.

\begin{prop}[Minimax separation for testing 2-SBM graphs]
\label{prop_sbm}
Consider the graph two-sample problem restricted to 2-SBM graphs with specified $n,m$ and $\cc$.
The minimax separation for operator norm is
\begin{displaymath}
 \sc^*_{op} \asymp n\cc \bigwedge \sqrt{\frac{n\cc}{m}} \qquad \text{for all } m\geq1. 
\end{displaymath}
For minimax separation in Frobenius norm $\sc^*_F$, the above rate hold only for $m\geq2$.
For $m=1$, we have loose bounds $n\cc \wedge \sqrt{n\cc} \lesssim \sc^*_F \lesssim n\cc$.

Hence, the minimax separation for 2-SBM is similar to the corresponding rates for $IER$ graphs (with possible exception of $\sc^*_F$ for $m=1$).
\end{prop}

\begin{proof}[Proof sketch]
    We first note that the testing problem is a restriction of the original $\IER$ testing problem, and hence, the upper bounds of $\sc^*_F$ and $\sc^*_{op}$, derived in Theorem~\ref{thm_frotest} and  Corollary~\ref{cor_optest_improved} respectively, also hold in this case. Hence, we only need to prove the lower bounds
    \begin{displaymath}
     \sc^*_{op} \gtrsim n\cc \bigwedge \sqrt{\frac{n\cc}{m}} \qquad \text{and} \qquad \sc^*_F \gtrsim n\cc \bigwedge \sqrt{\frac{n\cc}{m}} 
    \end{displaymath}
    for all $m\geq1$.
    This lower bound on $\sc^*_{op}$ follows directly from the proof of Lemma~\ref{thm_optest_LB}. Recall that the construction used to prove Lemma~\ref{thm_optest_LB} was essentially a case of distinguishing a 2-SBM from an ER, where the latter is also a special case of a 2-SBM.
    Hence, the same proof works in the present case as well, and the claimed lower bound for $\sc^*_{op}$ holds.
    
    The lower bound for $\sc^*_F$ also follows from the same construction and computing the Frobenius norm distance between the choice of population adjacency matrices used in proof of Lemma~\ref{thm_optest_LB}.
\end{proof}

Proposition~\ref{prop_sbm} may have important consequences since it apparently implies that for any model that is more complex than the simple 2-SBM, the two-sample problem is as difficult as the general setting of $\IER$ graphs.
However, a more in-depth study may be required before a strong claim can be made in this context.
For instance, one should take into account the fact that the broad literature on graph clustering and stochastic block model often require an additional assumption of balanced community size. It would be interesting to understand whether the presence of balanced communities also simplify the detection / testing problems.

In a broader context, one should note that Propositions~\ref{prop_er} and~\ref{prop_sbm} only provide minimax rates for the Frobenius and operator norm distances for these special classes of $\IER$ graphs. 
On the other hand, Proposition~\ref{prop_tv} and~\ref{prop_skl} demonstrate the limitation of distribution based distances when no restriction is imposed on $\IER$ model.
Hence, there is a possibility that, under special models, total variation and $KL$-divergence based tests are as useful or even better than Frobenius or operator norm based tests.
Insights into these questions, along with minimax rates for related models such as $k$-SBM and random dot product graphs, may provide a clear understanding of the problem of testing graphs on a common set of vertices.

\subsection{Extensions}

Several extensions of the two sample problem~\eqref{eqn_prob}--\eqref{eqn_prob_sets} can be studied.
Earlier in this section, we have discussed the possibility of considering alternative notions of sparsity.
Another interesting, and practically significant, extension is to the case of directed graphs.
The problem naturally extends to this framework, and the proposed adaptive tests easily tackle this generalisation without any critical modification.
For instance, in the case $\Psi_F$, one merely needs to define $\mhat$ and $\shat$ as a summation over all off-diagonal terms and the thresholds change only by constant factors.
The analysis of such tests as well as the minimax lower bounds can be easily derived from our proofs.
The same conclusion is true for the case of operator norm separation, particularly, when the upper bounds are derived based on $\Psi_{op}$ and the matrix Bernstein inequality.

In this paper, we only consider the problem of identity testing, that is, $P=Q$ or $d(P,Q)>\sc$.
One may also study the more general problem of closeness testing, which ignores small differences between the models, that is, one tests between the hypotheses
\begin{align*}
\Omega_0 &= \{ d(P,Q) \leq \epsilon,  ~\Vert P\Vert_{max} \vee \Vert Q\Vert_{max} \leq \cc\}, \qquad\textup{and}
\\
\Omega_1 &= \{  d( P,Q) > \sc, ~\Vert P\Vert_{max} \vee \Vert Q\Vert_{max} \leq \cc\}. 
\end{align*}
for some pre-specified $\epsilon<\sc$.
The proposed tests, which are primarily based on the principle of estimating $d(P,Q)$ may be easily adapted to this setting by appropriately modifying the test thresholds.
However, it is not clear whether the minimax separation bounds in Theorems~\ref{thm_frotest} and~\ref{thm_optest} easily extend to this setting as well.

From a practical perspective, one may face a more general problem of two sample graph testing, where the graphs are not defined on a common set of vertices and may even be of different sizes.
This situation is generally hard to study, but tests for this problem are often used in many applications, where one typically computes some scalar or vector function from each graph and comments on the difference between two graph populations based on this function~\citep{Stam_2007_jour_CerebralCortex}.
We study this principle in a recent work~\citep{Ghoshdastidar_2017_conf_COLT}, and propose a formal framework for testing between any two random graphs through the means of a network function $f : \mathcal{G}_{\geq n} \to \mathcal{M}$ that maps the space of graphs on at least $n$ vertices to some metric space $\mathcal{M}$.
We argue that if the network function concentrates for some sub-class of random graphs as $n\to\infty$, then one can indeed construct two sample tests based on the network function.
However, such a test cannot distinguish between equivalence classes, that is, random graph models that behave identically under the mapping $f$.

\appendix


\section{Proof of results in Section~3}

Before presenting the proofs of Propositions~\ref{prop_tv}--\ref{prop_zero}, we briefly recall the general technique for proving lower bounds in the minimax setting.
Throughout the appendix, we often denote a generic tuple $(P,Q)$ by $\theta$, and use $\omega$ to denote a generic population of $2m$ graphs $G_1,\ldots,G_m,H_1,\ldots,H_m$.

Consider the two sample problem in~\eqref{eqn_prob}--\eqref{eqn_prob_sets}.
Let $\theta_0\in\Omega_0$ be a particular instance satisfying the null hypothesis, and $\Theta_1\subset\Omega_1$ be a finite collection of instances satisfying $\Hc_1$. 
We specify $\theta_0$ and $\Theta_1$ later for each instance of the problem, but to prove a general lower bound, let $\theta_1$ be uniformly selected from $\Theta_1$.
The minimax risk~\eqref{eqn_minimaxrisk} can be bounded from below as
\begin{align*}
R^*(n,m,d,\sc,\cc) 
&\geq 
\inf_\Psi \left(\P_{\theta_0}(\Psi = 1) +  \sup_{\theta\in \Theta_1} \P_\theta(\Psi = 0)\right)
\\&\geq 
\inf_\Psi \big(\P_{\theta_0}(\Psi = 1) +  \E_{\theta_1\sim \text{Unif}(\Theta_1)} [\P_{\theta_1}(\Psi = 0)]\big)
\\&= 
1 + \inf_\Psi \big(\P_{\theta_0}(\Psi = 1) -  \E_{\theta_1\sim \text{Unif}(\Theta_1)} [\P_{\theta_1}(\Psi = 1)]\big)
\\&\geq 
1 - \sup_\Psi \big|\P_{\theta_0}(\Psi = 1) -  \E_{\theta_1\sim \text{Unif}(\Theta_1)} [\P_{\theta_1}(\Psi = 1)]\big| \;.
\end{align*}
Let $\mathcal{F}$ be the collection of all possible sets of $2m$ graphs on $n$ vertices, and let $F_\Psi\subset\mathcal{F}$ be the sub-collection of those instances for which $\Psi=1$.
Then, we can re-write above lower bound as
\begin{align*}
R^*(n,m,d,\sc,\cc)&\geq 
1 - \sup_{F_\Psi} \big|\P_{\theta_0}(F_\Psi) -  \E_{\theta_1\sim \text{Unif}(\Theta_1)} [\P_{\theta_1}(F_\Psi)]\big|
\\&\geq 
1 - \sup_{F\subset\mathcal{F}} \big|\P_{\theta_0}(F) -  \E_{\theta_1\sim \text{Unif}(\Theta_1)} [\P_{\theta_1}(F)]\big| 
\\&= 
1 - \frac12\sum_{\omega\in\mathcal{F}} \big|\P_{\theta_0}(\omega) -  \E_{\theta_1\sim \text{Unif}(\Theta_1)} [\P_{\theta_1}(\omega)]\big| 
\\&\geq
1 - \frac12\sqrt{\sum\limits_{\omega\in\mathcal{F}} \frac{(\E_{\theta_1\sim \text{Unif}(\Theta_1)} [\P_{\theta_1}(\omega)])^2}{\P_{\theta_0}(\omega)} -1}\;,
\end{align*}
Here, $\omega\in\mathcal{F}$ corresponds to a collection of $2m$ graphs. 
The equality follows by observing that both $\P_{\theta_0}(\cdot)$ and $\E_{\theta_1\sim \text{Unif}(\Theta_1)} [\P_{\theta_1}(\cdot)]$ define two measures on $\mathcal{F}$, and hence, the equality is due to equivalence of two definitions of total variation distance.
The last step is a consequence of Cauchy-Schwarz inequality along with the observation that $\sum_\omega \P_\theta(\omega)=1$ for any $\theta$.
Thus, to show that the minimax risk is larger than any $\eta\in(0,1)$, it suffices to show that for some $\theta_0\in\Omega_0$ and $\Theta_1\subset\Omega_1$,
\begin{equation}
\sum\limits_{\omega\in\mathcal{F}} \frac{(\E_{\theta_1\sim \text{Unif}(\Theta_1)} [\P_{\theta_1}(\omega)])^2}{\P_{\theta_0}(\omega)} \leq 1 + 4(1-\eta)^2.
\label{eqn_likeratio}
\end{equation}

\subsection{Proof of Proposition~\ref{prop_tv}}
The upper bound holds trivially.
To prove the lower bound, we consider the following choice of $\theta_0$ and $\Theta_1$.
Let $p=\frac{\cc}{2}$, and $\gamma\in(0,p]$.
Define $\theta_0=(P,Q)$ such that every off-diagonal entry in $P$ and $Q$ equals $p$, that is, both models correspond to Erd\H{o}s-R{\'e}nyi graphs with edge probability $p$.
Let $\Theta_1$ be the collection of all $\theta=(P,Q)$, where $P$ is same as before, but each off-diagonal entry in $Q$ is either $(p+\gamma)$ or $(p-\gamma)$.
Note that due to the symmetry of $Q$, there are exactly $2^{n(n-1)/2}$ elements in $\Theta_1$.

Moreover, $\Theta_1\subset \Omega_1$ if  $\gamma$ is such that $TV(\IER(P),\IER(Q))>\sc$ for all $(P,Q)\in\Theta_1$.
We first derive a condition on $\gamma$ so that this holds. 
We can characterise any $Q$ above as $Q_\epsilon = P + \gamma \epsilon$, where $\epsilon\in\R^{n\times n}$ is symmetric with zero diagonal and off-diagonal entries as $\pm1$. Recall that
\begin{align*}
TV(\IER(P),\IER(Q_\epsilon)) \geq H^2(\IER(P),\IER(Q_\epsilon))  = 1 - \sum_A \sqrt{P(A)}\sqrt{Q_\epsilon(A)} \,,
\end{align*}
where $H^2(\cdot,\cdot)$ is the squared Hellinger distance. The summation is over all possible adjacency matrices and we use $P(A)$ to denote the probability mass at $A$ under the distribution $\IER(P)$. 
We can write the above relation as
\begin{align*}
1- &TV\big(\IER(P),\IER(Q_\epsilon)\big)
\\&\leq \sum_A \prod_{i<j} \sqrt{p^{A_{ij}}(1-p)^{1-A_{ij}}}\sqrt{(p+\gamma\epsilon_{ij})^{A_{ij}}(1-p-\gamma\epsilon_{ij})^{1-A_{ij}}}
\\&=  \prod_{i<j} \sum_{A_{ij}\in\{0,1\}} \sqrt{p^{A_{ij}}(1-p)^{1-A_{ij}}}\sqrt{(p+\gamma\epsilon_{ij})^{A_{ij}}(1-p-\gamma\epsilon_{ij})^{1-A_{ij}}}
\\&=  \prod_{i<j} \sum_{A_{ij}\in\{0,1\}} \sqrt{p^{A_{ij}}(1-p)^{1-A_{ij}}}\sqrt{(p+\gamma\epsilon_{ij})^{A_{ij}}(1-p-\gamma\epsilon_{ij})^{1-A_{ij}}}
\\&= \prod_{i<j} \left[ \sqrt{p(p+\gamma\epsilon_{ij})} + \sqrt{(1-p)(1-p-\gamma\epsilon_{ij})}\right]
\\&= \prod_{i<j} \left[ p\sqrt{1+\frac{\gamma\epsilon_{ij}}{p}} + (1-p)\sqrt{1-\frac{\gamma\epsilon_{ij}}{1-p}}\right].
\end{align*}

On the third line, we swap the summation and the product using the relation $\sum_{i_1,\ldots,i_k} a_{i_1}\ldots a_{i_k} = \left( \sum_{i_1} a_{i_1}\right)\ldots\left( \sum_{i_k} a_{i_k}\right)$.
At this stage, we need the  upper bounds $\sqrt{1+x} \leq 1+\frac{x}{2} - \frac{x^2}{16}$ and $\sqrt{1-x} \leq 1-\frac{x}{2} - \frac{x^2}{16}$ for all $x\in[-1,1]$, which can be easily verified by squaring both sides.
Using these bounds, we can write
\begin{align*}
1- &TV\big(\IER(P),\IER(Q_\epsilon)\big)
\\&\leq \prod_{i<j} \left[ p\left(1+ \frac{\gamma\epsilon_{ij}}{2p} - \frac{\gamma^2}{16p^2}\right) + (1-p)\left(1- \frac{\gamma\epsilon_{ij}}{2(1-p)} - \frac{\gamma^2}{16(1-p)^2}\right)\right]
\\&=\left(1-\frac{\gamma^2}{16p(1-p)}\right)^{\binom{n}{2}}
\leq \exp\left(-\frac{n^2\gamma^2}{32\cc}\right)
\end{align*}
Thus, $(P,Q_\epsilon)\in\Omega_1$ for every $\epsilon$ if $\gamma > \frac{1}{n}\sqrt{32\cc\ln(\frac{1}{1-\sc})}$.
Let $\sc = 1 -\frac1n$ and $\gamma = \frac{8\sqrt{\cc \ln n}}{n}$.
The condition on $\cc$ stated in the proposition ensures that $\gamma \leq\frac\cc2$, and hence, $Q_\epsilon$ is non-negative and $\Vert Q_\epsilon\Vert_{max} \leq \cc$.

We now rely on the general technique for deriving lower bounds, and compute the quantity in~\eqref{eqn_likeratio} in the present case.
Let $\omega\in\mathcal{F}$ be the tuple $\omega = (G_1,\ldots,G_m,H_1,\ldots,H_m)$, where we assume that the first $m$ graphs are generated from the first model, and the rest from the second model.
Then 
\begin{align*}
\P_{\theta_0}(\omega) = \prod_{i<j} p^{(S_G)_{ij} + (S_H)_{ij}}(1-p)^{2m - (S_G)_{ij} - (S_H)_{ij}}\;,
\end{align*}
where $S_G = \sum_k A_{G_k}$ and $S_H = \sum_k A_{H_k}$.
On the other hand, by construction, every element in $\Theta_1$ is characterised by $\epsilon \in \{\pm1\}^{n(n-1)/2}$, which specifies if $Q_{ij} = (p+\gamma)$ or $(p-\gamma)$. 
Denoting the element by $\theta_\epsilon$, we have
\begin{align*}
\P_{\theta_\epsilon}(\omega) = \prod_{i<j}p^{(S_G)_{ij}} (1-p)^{m -  (S_G)_{ij}}
(p+\epsilon_{ij}\gamma)^{(S_H)_{ij}}(1-p-\epsilon_{ij}\gamma)^{m - (S_H)_{ij}}\;.
\end{align*}
Based on this, one can compute the quantity in~\eqref{eqn_likeratio} as
\begin{align*}
\sum\limits_{\omega\in\mathcal{F}} \frac{(\E_{\theta_1\sim \text{Unif}(\Theta_1)} [\P_{\theta_1}(\omega)])^2}{\P_{\theta_0}(\omega)}
=& \frac{1}{2^{n(n-1)}}\sum_\omega\sum_{\epsilon,\epsilon'} \prod_{i<j}
\frac{p^{(S_G)_{ij}}(1-p)^{m - (S_G)_{ij}}}{p^{(S_H)_{ij}}(1-p)^{m - (S_H)_{ij}}}
\\&\times(p+\epsilon_{ij}\gamma)^{(S_H)_{ij}}(1-p-\epsilon_{ij}\gamma)^{m - (S_H)_{ij}}
\\&\times(p+\epsilon'_{ij}\gamma)^{(S_H)_{ij}}(1-p-\epsilon'_{ij}\gamma)^{m - (S_H)_{ij}}
\end{align*}
Pushing the summation over $\omega$ inside the product and transforming it to a summation over every $i<j$, the above quantity corresponds to summing over possible values of $(S_G)_{ij}$ and $(S_H)_{ij}$, where each of them can take the value $k$ in $\binom{m}{k}$ ways. 
We obtain
\begin{align*}
&\sum\limits_{\omega\in\mathcal{F}} \frac{(\E_{\theta_1\sim \text{Unif}(\Theta_1)} [\P_{\theta_1}(\omega)])^2}{\P_{\theta_0}(\omega)}
\\&= \frac{1}{2^{n(n-1)}}\sum_{\epsilon,\epsilon'}\prod_{i<j} \sum_{k_G,k_H=0}^m 
\binom{m}{k_G}\binom{m}{k_H} p^{k_G} (1-p)^{m-k_G}
\\&~~\times \left(p+(\epsilon_{ij}+\epsilon'_{ij})\gamma + \frac{\epsilon_{ij}\epsilon'_{ij}\gamma^2}{p} \right)^{k_H}
\left(1-p-(\epsilon_{ij}+\epsilon'_{ij})\gamma + \frac{\epsilon_{ij}\epsilon'_{ij}\gamma^2}{1-p} \right)^{m-k_H}
\end{align*}
Note that this swapping summation and product follows from the same general relation used earlier.
One can now separate the terms corresponding to $k_G$ and $k_H$, and check that the former sums to 1 due to binomial expansion, while the latter sums to $\left(1+ \frac{\epsilon_{ij}\epsilon'_{ij}\gamma^2}{p(1-p)}\right)^m$. 
Thus, we have
\begin{align}
\sum\limits_{\omega\in\mathcal{F}} &\frac{(\E_{\theta_1\sim \text{Unif}(\Theta_1)} [\P_{\theta_1}(\omega)])^2}{\P_{\theta_0}(\omega)}
\nonumber
\\&= \frac{1}{2^{n(n-1)}} \sum_{\epsilon,\epsilon'} \prod_{i<j} \left(1+ \frac{\epsilon_{ij}\epsilon'_{ij}\gamma^2}{p(1-p)}\right)^m
\label{eqn_frotest_LB_prf0}
\\&=  \prod_{i<j}\left[\frac14 \sum_{\epsilon_{ij},\epsilon'_{ij}\in\{\pm1\}} \left(1+ \frac{\epsilon_{ij}\epsilon'_{ij}\gamma^2}{p(1-p)}\right)^m\right]
\nonumber
\\&= \left[\frac{1}{2}\left(1 + \frac{\gamma^2}{p(1-p)}\right)^m + \frac12\left(1 -\frac{\gamma^2}{p(1-p)}\right)^m\right]^{\binom{n}{2}}
\label{eqn_frotest_LB_prf1}
\end{align}
where in the second step we again swap the product and summation.
Finally, using the facts that $(1+x) \leq \exp(x)$ and $\cosh(x)\leq \exp(x^2/2)$ for all $x$, which can be verified from Taylor series expansion, we have
\begin{align}
\sum\limits_{\omega\in\mathcal{F}} \frac{(\E_{\theta_1\sim \text{Unif}(\Theta_1)} [\P_{\theta_1}(\omega)])^2}{\P_{\theta_0}(\omega)}
\leq&\left[\frac{1}{2}\exp\left(\frac{m\gamma^2}{p(1-p)}\right) + \frac12\exp\left(-\frac{m\gamma^2}{p(1-p)}\right)\right]^{\frac{n^2}{2}}
\nonumber
\\=& \left[\cosh\left(\frac{m\gamma^2}{p(1-p)}\right)\right]^{n^2/2}
\nonumber
\\\leq& \exp\left(\frac{4n^2m^2\gamma^4}{\cc^2}\right)\;.
\label{eqn_frotest_LB_prf2}
\end{align}
We observe that~\eqref{eqn_likeratio} holds for $\gamma \leq \sqrt{\frac{\ell_\eta \cc}{2mn}}$,
where $\ell_\eta = \sqrt{\ln(1+4(1-\eta)^2)}$.
Using the value for $\gamma = \frac{8\sqrt{\cc\ln n}}{n}$ assumed earlier, one can see that both~\eqref{eqn_likeratio} and $\Theta_1\subset \Omega_1$ hold under the stated condition on $m$, which leads to the claimed lower bound that $\sc^* > 1 - \frac1n$.

\subsection{Proof of Proposition~\ref{prop_skl}}
We begin by computing the symmetric KL-divergence between two IER models. Observe that
\begin{align*}
KL&(\IER(P)\Vert\IER(Q)) 
= \sum_A P(A)\ln\left(\frac{P(A)}{Q(A)}\right) 
\\&=\sum_A \sum_{i<j} P(A)  \ln \left( \frac{P_{ij}^{A_{ij}} (1-P_{ij})^{1-A_{ij}}}{Q_{ij}^{A_{ij}} (1-Q_{ij})^{1-A_{ij}}}\right)
\\&= \sum_{i<j} \sum_A \left( \prod_{i'<j'}  P_{i'j'}^{A_{i'j'}} (1-P_{i'j'})^{1-A_{i'j'}} \right) \ln \left( \frac{P_{ij}^{A_{ij}} (1-P_{ij})^{1-A_{ij}}}{Q_{ij}^{A_{ij}} (1-Q_{ij})^{1-A_{ij}}}\right).
\end{align*}
We now swap the summation and product, and note that there are two cases --- if $(i',j')\neq (i,j)$, then $\sum_{A_{i'j'}}P_{i'j'}^{A_{i'j'}} (1-P_{i'j'})^{1-A_{i'j'}} = 1$, whereas if $(i',j')=(i,j)$ the log term needs to be counted.
Hence,
\begin{align*}
KL(\IER(P)\Vert\IER(Q)) = \sum_{i<j} P_{ij}\ln\left(\frac{P_{ij}}{Q_{ij}}\right) + (1-P_{ij}) \ln\left(\frac{1-P_{ij}}{1-Q_{ij}}\right),
\end{align*}
which in turn implies
\begin{align*}
SKL(\IER(P),\IER(Q)) &= KL(\IER(P)\Vert\IER(Q)) + KL(\IER(Q)\Vert\IER(P))
\\&=\sum_{i<j} (P_{ij}-Q_{ij})\ln\left(\frac{P_{ij}(1-Q_{ij})}{Q_{ij}(1-P_{ij})}\right).
\end{align*}
Our proof of Proposition~\ref{prop_skl} relies the fact that $SKL(\IER(P),\IER(Q)) = \infty$ when there is $i<j$ such that $P_{ij}\neq0$ and $Q_{ij}=0$, or vice versa, which can be easily verified from the above expression. 
Based on this, we consider the following choice of $\theta_0$ and $\Theta_1$.
Let $p=\frac{\cc}{2}$, and define $\theta_0=(P,Q)$ such that every off-diagonal entry in $P$ and $Q$ equals $p$.
Let $\Theta_1$ be the collection of all $\theta=(P,Q)$, where $P$ is the same as before, and $Q$ equals $P$ on all entries except one off-diagonal entry which is zero.
As in the proof of Proposition~\ref{prop_tv}, we may denote any such $Q$ as $Q_\epsilon = P + p\epsilon$, where $\epsilon \in\R^{n\times n}$ is symmetric and zero everywhere except one entry above diagonal where $\epsilon_{ij} = -1$.
Note that there are $\binom{n}{2}$ such $\epsilon$'s and hence, this is the size of $\Theta_1$.
It is also obvious that $SKL(\IER(P),\IER(Q_\epsilon)) = \infty$.

We now compute the term in~\eqref{eqn_likeratio} in this case. Note that the computation up to~\eqref{eqn_frotest_LB_prf0} can be done for any matrix $\epsilon$, and in the present case, the same computation yields
\begin{align*}
\sum\limits_{\omega\in\mathcal{F}} \frac{(\E_{\theta_1\sim \text{Unif}(\Theta_1)} [\P_{\theta_1}(\omega)])^2}{\P_{\theta_0}(\omega)}
&= \frac{1}{\binom{n}{2}^2} \sum_{\epsilon,\epsilon'} \prod_{i<j} \left(1+ \frac{\epsilon_{ij}\epsilon'_{ij}p^2}{p(1-p)}\right)^m
\end{align*}
Observe that if $\epsilon\neq \epsilon'$ then $\epsilon_{ij}\epsilon'_{ij}=0$ for all $i<j$, and the product is one in these cases.
If $\epsilon=\epsilon'$, then the product equals $(1+\frac{p}{1-p})^m \leq (1+\cc)^m \leq \exp(m\cc)$. Hence,
\begin{align*}
\sum\limits_{\omega\in\mathcal{F}} \frac{(\E_{\theta_1\sim \text{Unif}(\Theta_1)} [\P_{\theta_1}(\omega)])^2}{\P_{\theta_0}(\omega)}
&\leq \frac{\binom{n}{2}-1}{\binom{n}{2}}. + \frac{\exp(m\cc)}{\binom{n}{2}} 
\leq 1+ \frac{4}{n^2}\exp(m\cc)\;.
\end{align*}
Thus,~\eqref{eqn_likeratio} is satisfied when $m\leq \frac{2}{\cc}\ln((1-\eta)n)$.

\subsection{Proof of Proposition~\ref{prop_zero}}
Consider the construction in the proof of Proposition~\ref{prop_tv}. If $\gamma>0$, then $\Vert P-Q_\epsilon\Vert_0 = n(n-1)$ for all $(P,Q_\epsilon)\in\Theta_1$.
But, as shown in the same proof,~\eqref{eqn_likeratio} holds for $\gamma \leq \sqrt{\frac{\ell_\eta \cc}{2mn}}$
Note that this upper bound on $\gamma$ is strictly positive for all $m<\infty$, $\cc>0$ and $\eta<1$.
Hence, the result follows.

\section{Proof of results in Section~4}

\subsection{Proof of Lemma~\ref{thm_frotest_LB}}

The necessary condition is proved using the instance constructed in the proof of Proposition~\ref{prop_tv}.
Observe that for $(P,Q_\epsilon)\in\Theta_1$, $\Vert P-Q_\epsilon\Vert_F = \gamma\sqrt{n(n-1)} > \frac{n\gamma}{2}$.

To prove condition (i) in the statement of the lemma, we use the bound in~\eqref{eqn_frotest_LB_prf2} to obtain that~\eqref{eqn_likeratio} holds for $\gamma \leq \sqrt{\frac{\ell_\eta \cc}{2mn}}$.
We now set $\gamma = \sqrt{\frac{\ell_\eta \cc}{2mn}} \wedge \frac\cc2$, where the second term ensures $\Vert Q_\epsilon\Vert_{max} \leq \cc$.
The claim follows from the fact that $\Theta_1\subset \Omega_1$ if $\sc < \frac{n\gamma}{2}$. 

To prove condition (ii) in the lemma, observe that the bound in~\eqref{eqn_frotest_LB_prf1} equals 1 for $m=1$, 
and hence,~\eqref{eqn_likeratio} is satisfied for all $\gamma \leq \frac\cc2$. Setting $\gamma = \frac\cc2$ leads to the claim.

\subsection{Proof of Lemma~\ref{thm_frotest_UB}}\label{Proof_thm_frotest_UB}

We assume that $m$ is even for convenience. If $m>2$ and odd, one may simply work with $m-1$ samples, and the result still holds with possible change in constants.
\paragraph{\textbf{Preparatory computations: Ingredients for Chebyshev's inequality}}
Let $\mu = \E[\mhat] = \frac{m^2}{8}\Vert P-Q\Vert_F^2$ and $\sigma^2 = \E[\shat^2] = \frac{m^2}{8}\Vert P+Q\Vert_F^2$. Since all terms in $\mhat$ and $\shat$ are independent, one can compute their variances as
\begin{align*}
\Var(\mhat) &= \sum_{i<j}  \Bigg[ \frac{m^2}{4}( P_{ij}(1-P_{ij}) + Q_{ij}(1-Q_{ij}))^2 
\\&\hskip20ex+  \frac{m^3}{4}(P_{ij}-Q_{ij})^2 ( P_{ij}(1-P_{ij}) + Q_{ij}(1-Q_{ij}) ) \Bigg]
\\&\leq \frac{m^2}{8}\sum_{ij} (P_{ij}+Q_{ij})^2 + \frac{m^3}{8} \sum_{ij}(P_{ij}-Q_{ij})^2(P_{ij}+Q_{ij})
\\&\leq \frac{m^2}{8} \Vert P+Q\Vert_F^2 + \frac{m^3}{8} \Vert P-Q\Vert_F^2 \Vert P+Q\Vert_F\;.
\end{align*}
In the last step, we use the Cauchy-Schwarz inequality for the second summation, followed by the fact $\Vert x\Vert_4 \leq \Vert x\Vert_2$.
Similarly, $\Var(\shat^2) \leq \frac{m^2}{8} \Vert P+Q\Vert_F^2 + \frac{m^3}{8} \Vert P+Q\Vert_F^3$, which is at most $\frac{m^3}{4}\Vert P+Q\Vert_F^3$ when $m\Vert P+Q\Vert_F \geq 1$.
\paragraph{\textbf{Controlling the type-I-error rate}}~\\
Consider $\theta=(P,Q)\in\Omega_0$. The test $\Psi_F$ makes an error only if both events in~\eqref{eqn_fro_test} occur, that is,
\begin{align}
\P_\theta(\Psi_F=1) \leq \P_\theta\left(\frac\mhat\shat > \frac{t_1}{\sqrt\eta}\right) \bigwedge \P_\theta \left(\shat > \frac{t_2}{\eta^{3/2}}\right)
\label{eqn_froUB_H0}
\end{align}
It is convenient to bound both probabilities separately based on the following case analysis:
\begin{description}
\item[\normalfont\underline{\textit{Case 1: $\Vert P+Q\Vert_{F}>\frac{C''}{\eta m}$, $C''\geq 1$:}}]~Here we bound the first term in~\eqref{eqn_froUB_H0} as
\begin{align*}
\P_\theta\left(\frac\mhat\shat > \frac{t_1}{\sqrt\eta}\right) 
&\leq \P_\theta\left(\mhat > \frac{\sigma t_1}{2\sqrt\eta}\right) + \P_\theta\left(\shat^2 < \frac{\sigma^2}{4}\right)
\\&\leq\frac{4\eta}{\sigma^2 t_1^2} \Var(\mhat) + \frac{16}{9\sigma^4} \Var(\shat^2)
\end{align*}
using the Chebyshev inequality. For $P=Q$, $\Var(\mhat) \leq \sigma^2$ and since $m\Vert P+Q\Vert_F \geq 1$, we have $\Var(\shat^2) \leq 2\sqrt{8}\sigma^3$. Thus, in this case, $\P_\theta(\Psi_F = 1) \leq \frac{4\eta}{t_1^2} + \frac{256\eta}{9C''}$, which is smaller than $\frac\eta2$ if $t_1$ and $C''$ are large enough.
\item[\normalfont\underline{\textit{Case 2: $\Vert P+Q\Vert_{F}\leq\frac{C''}{\eta m}$:}}]~Here the error probability in~\eqref{eqn_froUB_H0} can be bounded using the Markov inequality.
\begin{align*}
\P_\theta(\Psi_F=1) \leq \P_\theta \left(\shat^2 > \frac{t_2^2}{\eta^{3}}\right) \leq \frac{\sigma^2\eta^3}{t_2^2} \leq \frac{(C'')^2\eta}{8t_2^2}\;,
\end{align*}
which can be made smaller than $\frac\eta2$ by choosing $t_2$ large.
\end{description}
Thus the Type-I error rate for $\Psi_F$ is at most $\frac\eta2$.
\paragraph{\textbf{Controlling the type-II-error rate}}~\\
For $\theta=(P,Q)\in\Omega_1$, we can bound the error rate as
\begin{align}
\P_\theta(\Psi_F=0) &\leq \P_\theta\left(\frac\mhat\shat \leq \frac{t_1}{\sqrt\eta}\right) + \P_\theta \left(\shat \leq \frac{t_2}{\eta^{3/2}}\right)
\nonumber
\\&\leq \P_\theta\left(\mhat \leq \frac{3\sigma t_1}{2\sqrt\eta}\right) + \P_\theta\left(\shat^2 \geq \frac{9\sigma^2}{4}\right)+ \P_\theta \left(\shat^2 \leq \frac{t_2^2}{\eta^{3}}\right).
\label{eqn_froUB_H1}
\end{align} 
We aim to bound each term by $\frac\eta6$ so that the Type-II error rate is at most $\frac\eta2$.\\
We use the relations $\Vert P+Q\Vert_F\leq 2n\cc$ and $\sc \geq \sqrt{\frac{Cn\cc}{\eta m}}$ to write 
\begin{equation*}
\mu = \frac{m^2}{8}\Vert P-Q\Vert_F^2 > \frac{m^2\sc^2}{8} \geq \frac{Cmn\cc}{8\eta} \geq \frac{\sqrt{2}C\sigma}{8\eta} \;.
\end{equation*}
Hence, by choosing $C$ large enough, we can assume that $\frac{2\sigma t_1}{\sqrt\eta} \leq \frac{2\sigma t_1}{\eta}\leq \mu$, and can bound the first term as
\begin{align*}
\P_\theta\left(\mhat \leq \frac{3\sigma t_1}{2\sqrt\eta}\right) \leq \P_\theta\left(\mu - \mhat \geq \frac{\mu}{4}\right)
&\leq \frac{16}{\mu^2}\Var(\mhat)
\leq \frac{16}{\mu^2}\left(\sigma^2 + \sqrt{8}\mu\sigma\right).
\end{align*}
Noting that $\sigma \leq \frac{\mu \eta}{2t_1}$, we can get the above bound smaller than $\frac\eta6$ by choosing $t_1$ large.\\
To bound the second and third term in~\eqref{eqn_froUB_H1}, we note that $\Vert P+Q\Vert_F \geq \Vert P-Q\Vert_F > \sc \geq \frac{C'}{\eta^{3/2}m}$, and so, $\sigma \geq \frac{C'}{\sqrt8\eta^{3/2}}$. For large $C'$, we can write $\frac{t_2}{\eta^{3/2}}  \leq \frac{\sigma}{2}$. 
Hence, we may bound each of the last two terms in~\eqref{eqn_froUB_H1} by
\begin{align*}
\P_\theta\left( |\shat^2 - \sigma^2| \geq \frac{3\sigma^2}{4}\right) 
\leq \frac{16}{9\sigma^4} \Var(\shat^2) \leq \frac{256\eta^{3/2}}{9C'}
\end{align*} 
which is at most $\frac\eta6$ for large $C'$.
In the last step, we use $\Var(\shat^2) \leq 2\sqrt{8}\sigma^3$ since $m\Vert P+Q\Vert_F\geq1$.

\subsection{Proof of Proposition~\ref{thm_frotest_UB2}} 

The proof is along the lines of the proof of Lemma~\ref{thm_frotest_UB}, but uses a stronger concentration inequality that helps in improving the dependence on $\eta$ at the cost of an additional $\ln n$ factor.
\paragraph{\textbf{Preparation: New concentration inequality}}~\\
We now state a generic concentration bound that can be applied to both $\mhat$~\eqref{eqn_mhat} and $\shat^2$~\eqref{eqn_shat} afterwards. 
\begin{lem}[Concentration inequality for sum of ``product of sums'']
\label{lem_froUB_conc2}
Let $m$ be even and $d$ be any positive integer. Let $\{X_{kl} : 1\leq k\leq m, 1\leq l \leq d\}$ be a collection of independent random variables with $\E[X_{kl}] = a_l$, $|X_{kl} - a_l| \leq 2$ almost surely, and $\Var(X_{kl}) \leq v_l$.

Let $S_l = \sum\limits_{k\leq m/2} X_{kl}$ and $S'_l = \sum\limits_{k> m/2} X_{kl}$.
For any $\epsilon \in(0,\frac15)$,
\begin{equation}
\P\left( \sum_l \bigg( S_lS'_l - \frac{m^2a_l^2}{4} \bigg) > \tau_1 + \tau_2\right) \leq 6d\epsilon,
\label{eqn_sps_conc}
\end{equation}
where $a = \sqrt{\sum_l a_l^2}$, $v = \sqrt{\sum_l v_l^2}$,
\begin{align*}
z &= \frac43\ln\left(\frac2\epsilon\right) + \sqrt{mv\ln\left(\frac2\epsilon\right)}, \\ 
\tau_1 &= mv\sqrt{\frac12\ln\left(\frac{1}{\epsilon d}\right)} + \frac{2z^2}{3}\ln\left(\frac{1}{\epsilon d}\right),  
\qquad \text{and}\\
\tau_2 &= ma\sqrt{\frac{mv}{2}\ln\left(\frac{1}{\epsilon d}\right)} +\frac{2maz}{3}\ln\left(\frac{1}{\epsilon d}\right). 
\nonumber
\end{align*}
A similar concentration inequality also holds for the lower tail probability.
\end{lem}

Lemma~\ref{lem_froUB_conc2} improves upon Bernstein's inequality with respect to the allowable range of parameters $m$ and $d$. This improvement is possible because unlike Bernstein's inequality that only assumes $|S_\ell S'_\ell| \leq \frac14 m^2$ with probability 1, we use the fact that $S_\ell S'_\ell$ is a product of sums.
In the context of Proposition~\ref{thm_frotest_UB2}, we exploit that each product term in $\mhat$ and $\shat$ is a product of two sums, and hence, the individual terms also concentrate.
We prove the Lemma~\ref{lem_froUB_conc2} later in the section.
\paragraph{\textbf{Target probabilities to be bounded}}~\\
We now start with the proof of Proposition~\ref{thm_frotest_UB2}.
Similar to~\eqref{eqn_froUB_H0} and~\eqref{eqn_froUB_H1}, we have that
\begin{align}
\label{eqn_froUB2_H0}
\P_\theta(\Psi'_F = 1) 
\leq \left(\P_\theta\left(\mhat > \frac{\sigma t_1}{2}\ln\left(\frac{2}{\eta}\right)\ln\left(\frac{n}{\eta}\right)\right) + \P_\theta\left(\shat^2 < \frac{\sigma^2}{4}\right) \right)&
\\
\nonumber
\bigwedge \P_\theta\left( \shat^2 > t_2^2 \ln^4 \left(\frac2\eta\right) \ln^2 \left(\frac{n}{\eta}\right) \right)&
\end{align}
for $\theta\in\Omega_0$, and 
\begin{align}
\label{eqn_froUB2_H1}
\P_\theta(\Psi'_F = 0) 
\leq \P_\theta\left(\mhat \leq \frac{3\sigma t_1}{2}\ln\left(\frac{2}{\eta}\right)\ln\left(\frac{n}{\eta}\right)\right) + \P_\theta\left(\shat^2 \geq \frac{9\sigma^2}{4}\right) &
\\
\nonumber
+~ \P_\theta\left( \shat^2 \leq t_2^2 \ln^4 \left(\frac2\eta\right) \ln^2 \left(\frac{n}{\eta}\right) \right)&
\end{align}
for $\theta\in\Omega_1$.\\
The claim of the proposition follows if each term in~\eqref{eqn_froUB2_H0} and~\eqref{eqn_froUB2_H1} is at most $\frac\eta6$, which we show using Lemma~\ref{lem_froUB_conc2}.
\paragraph{\textbf{Controlling the type-II-error rate}}~\\
Firstly, let $\theta\in\Omega_1$. Observe that $\sigma = \frac{m}{\sqrt{8}}\Vert P+Q\Vert_F \leq {mn\cc}$.
Furthermore, due to the stated condition on $\sc$,
\begin{align}
\label{eqn_froUB2_pf1}
\mu &> \frac{m^2\sc^2}{8} \geq \frac{C^2}{8} mn\cc\ln^2\left(\frac2\eta\right) \ln\left(\frac{n}{\eta}\right) \geq \frac{C^2}{8} \sigma \ln^2\left(\frac2\eta\right) \ln\left(\frac{n}{\eta}\right), 
\\
\label{eqn_froUB2_pf2}
\text{and~~}\sigma &> \frac{m\sc}{\sqrt{8}} \geq \frac{C'}{\sqrt{8}}\ln^2\left(\frac2\eta\right) \ln\left(\frac{n}{\eta}\right).
\end{align}
Hence, for $C,C'$ large enough with respect to $t_1,t_2$, we may write~\eqref{eqn_froUB2_H1} as
\begin{align}
\label{eqn_froUB2_H1_1}
\P_\theta(\Psi'_F = 0) 
\leq \P_\theta&\left(\mhat \leq \frac{\mu}{2}\right) 
+ \P_\theta\left(\shat^2\geq \frac{3\sigma^2}{2}\right) 
+ \P_\theta\left(  \shat^2 \leq \frac{\sigma^2}{2}\right).
\end{align}
To deal with the first term, we use Lemma~\ref{lem_froUB_conc2} with $X_{kl} = (A_{G_k})_{ij} - (A_{H_k})_{ij}$, where each $l$ corresponds to an entry $(i,j)$ and $d=\binom{n}{2}$. 
We may set $v_l = P_{ij}+Q_{ij}$, and so, $v = \frac{1}{\sqrt2}\Vert P+Q\Vert_F = \frac{2\sigma}{m}$ and $a = \frac{2}{m}\sqrt{\mu}$.
We also let $\epsilon = \frac{\eta}{18n^2}$ so that the probability bound in Lemma~\ref{lem_froUB_conc2} is at most $\frac\eta6$, which also implies $\ln(\frac2\epsilon) < 7\ln(\frac{n}{\eta})$ and $\ln(\frac{1}{\epsilon d}) < 6\ln(\frac2\eta)$ for all $\eta<1$ and $n\geq2$. Substituting these and using~\eqref{eqn_froUB2_pf2}, we observe that $z$ in~\eqref{eqn_sps_conc} is at most $z\leq c_0\sqrt{\sigma \ln(\frac{n}{\eta})}$ for some constant $c_0$.
And so, $\tau_1+\tau_2$ in~\eqref{eqn_sps_conc} satisfies
\begin{align*}
\tau_1+\tau_2
&\lesssim mv\sqrt{\ln\left(\frac{2}{\eta}\right)} \bigvee \sigma \ln\left(\frac{2}{\eta}\right)  \ln\left(\frac{n}{\eta}\right) \bigvee ma\sqrt{mv\ln\left(\frac{2}{\eta}\right)} 
\\&\hskip35ex\bigvee ma \ln\left(\frac{2}{\eta}\right) \sqrt{\sigma \ln\left(\frac{n}{\eta}\right)}
\\&\lesssim \sigma \ln\left(\frac{2}{\eta}\right)  \ln\left(\frac{n}{\eta}\right) \bigvee \ln\left(\frac2\eta\right) \sqrt{ \mu\sigma \ln\left(\frac{n}{\eta}\right)}.
\end{align*}
 Based on the relation between $\mu$ and $\sigma$ in~\eqref{eqn_froUB2_pf1}, we conclude that the above quantity is smaller than $\frac\mu2$ for large enough $C$.
Hence, due to Lemma~\ref{lem_froUB_conc2}, the first term in~\eqref{eqn_froUB2_H1_1} is bounded by $4d\epsilon\leq \frac\eta6$.

The second and third terms in~\eqref{eqn_froUB2_H1_1} are bounded in similar way, but now we let $X_{kl} = (A_{G_k})_{ij} + (A_{H_k})_{ij}$. So, we have $a = \frac{1}{\sqrt2}\Vert P+Q\Vert_F = \frac{2\sigma}{m}$, and may again set $v_l = P_{ij}+Q_{ij}$, which gives $v = a = \frac{2\sigma}{m}$.
For $\epsilon = \frac{\eta}{18n^2}$, we again have $z\leq c_0\sqrt{\sigma \ln(\frac{n}{\eta})}$ using~\eqref{eqn_froUB2_pf2}, and
\begin{align*}
\tau_1+\tau_2 
\lesssim \left(\sigma \ln\left(\frac{2}{\eta}\right)  \ln\left(\frac{n}{\eta}\right) \bigvee \sigma^{3/2} \ln\left(\frac2\eta\right) \sqrt{  \ln\left(\frac{n}{\eta}\right)}\right)
\lesssim \frac{\sigma^2}{\sqrt{C'}},
\end{align*}
which is smaller than $\frac{\sigma^2}{2}$ for large enough $C'$.
Based on this and Lemma~\ref{lem_froUB_conc2}, the second and third terms in~\eqref{eqn_froUB2_H1_1} are at most $\frac\eta6$. 
\paragraph{\textbf{Controlling the type-I-error rate}}~\\
We now consider $\theta\in\Omega_0$, where $\mu=0$. As in Section \ref{Proof_thm_frotest_UB}, it is convenient to distinguish two cases:
\begin{description}
\item[\normalfont\underline{\textit{Case 1: $\sigma\geq \xi \ln^2(\frac2\eta)\ln(\frac{n}{\eta})$, $\xi> 1$:}}]~We first assume $\sigma\geq \xi \ln^2(\frac2\eta)\ln(\frac{n}{\eta})$ for some $\xi>1$ to be specified later.
Under this assumption, we show that the first two terms in~\eqref{eqn_froUB2_H0} are smaller than $\frac\eta6$.
For the first term, we invoke Lemma~\ref{lem_froUB_conc2} with $X_{kl} = (A_{G_k})_{ij} - (A_{H_k})_{ij}$ and observe that $a_l=0$, $v_l = P_{ij}+Q_{ij}$. Hence, $a=0$ and $v = \frac{2\sigma}{m}$.
Let $\epsilon = \frac{\eta}{12n^2}$.
Due to our assumption on $\sigma$, we have $\sigma \geq \ln(\frac{n}{\eta})$ and so $z\leq c_0 \sqrt{\sigma \ln(\frac{n}{\eta})}$ for some constant $c_0$. 
We also have, for some constant $c_1$,
\begin{align*}
\tau_1 \leq c_1 \left(mv\sqrt{\ln\left(\frac{2}{\eta}\right)} \bigvee \sigma \ln\left(\frac{2}{\eta}\right)  \ln\left(\frac{n}{\eta}\right)\right)
= c_1 \sigma \ln\left(\frac{2}{\eta}\right)  \ln\left(\frac{n}{\eta}\right)
\end{align*}
and $\tau_2 =0$. Hence, by setting $t_1\geq c_1$, we can bound the tail probability by $\frac\eta6$.
To bound the second term in~\eqref{eqn_froUB2_H0} under the assumption $\sigma\geq \xi \ln^2(\frac2\eta)\ln(\frac{n}{\eta})$, we take the approach similar to the concentration of $\shat$ under the alternative hypothesis.
Let $X_{kl} = (A_{G_k})_{ij} + (A_{H_k})_{ij}$ and so, $a = v = \frac{2\sigma}{m}$.
As in the case of the alternative hypothesis, we have
\begin{align*}
\tau_1+\tau_2 
&\lesssim  \left(\sigma \ln\left(\frac{2}{\eta}\right)  \ln\left(\frac{n}{\eta}\right) \bigvee \sigma^{3/2} \ln\left(\frac2\eta\right) \sqrt{  \ln\left(\frac{n}{\eta}\right)}\right)
\lesssim \frac{\sigma^2}{\sqrt\xi} .
\end{align*}
For $\xi$ large enough, we have $\tau_1+\tau_2 \leq \frac{3\sigma^2}{4}$ and so, the second term  in~\eqref{eqn_froUB2_H0} is at most $\frac\eta6$. Thus, $\P_\theta(\Psi'_F=1)\leq \frac\eta3$ when $\sigma\geq \xi \ln^2(\frac2\eta)\ln(\frac{n}{\eta})$.

\item[\normalfont\underline{\textit{Case 2: $\sigma< \xi \ln^2(\frac2\eta)\ln(\frac{n}{\eta})$:}}]~
We show that the third term in~\eqref{eqn_froUB2_H0} is smaller than $\frac\eta6$.
Observe that if $t_2 > \xi$,
\begin{align*}
\P_\theta&\left( \shat^2 > t_2^2 \ln^4 \left(\frac2\eta\right) \ln^2 \left(\frac{n}{\eta}\right) \right)
\\&\leq  \P_\theta\left( \shat^2 - \sigma^2 > (t_2^2-\xi^2) \ln^4 \left(\frac2\eta\right) \ln^2 \left(\frac{n}{\eta}\right) \right),
\end{align*}
which we bound using Lemma~\ref{lem_froUB_conc2}.
Define $X_{kl} = (A_{G_k})_{ij} + (A_{H_k})_{ij}$ and $\epsilon = \frac{\eta}{18n^2}$ as above, and observe that under the condition on $\sigma$, we have $z \leq \sqrt{\xi}\ln(\frac2\eta)\ln(\frac{n}{\eta})$ noting that $\xi>1$ and $v=a= \frac{2\sigma}{m}$. Hence,
\begin{align*}
\tau_1+\tau_2
&\lesssim \sigma\sqrt{\ln\left(\frac{2}{\eta}\right)} \bigvee \xi \ln^3\left(\frac{2}{\eta}\right)  \ln\left(\frac{n}{\eta}\right) 
\\&\qquad\qquad\qquad\qquad\bigvee \sigma^{3/2}\sqrt{\ln\left(\frac{2}{\eta}\right)} \bigvee \sigma\sqrt{\xi} \ln^2\left(\frac{2}{\eta}\right)  \ln\left(\frac{n}{\eta}\right)
\\&\lesssim \xi^{3/2} \ln^4\left(\frac{2}{\eta}\right)  \ln^2\left(\frac{n}{\eta}\right),
\end{align*}
which is smaller than the above threshold for $t_2^2 \geq 2\xi^2$ and so, the above probability is smaller than $\frac\eta6$.
\end{description}
Hence, for $t_1$ and $t_2$ large enough, the rejection rate under null is also smaller than $\frac\eta2$, which completes the proof of the stated result. We conclude with the proof of Lemma~\ref{lem_froUB_conc2}.

\begin{proof}[Proof of Lemma~\ref{lem_froUB_conc2}]
We only prove the bound on upper tail probability as the corresponding result for the lower tail probability can be proved in a similar way.
We define the events $\xi_l = \left\{\left|S_l - \frac{ma_l}{2}\right| \leq z\right\}$ and  $\xi'_l = \left\{\left|S'_l - \frac{ma_l}{2}\right| \leq z\right\}$ for $l=1,\ldots,d$, and let $\xi = \bigcap\limits_l (\xi_l \cap \xi'_l)$. 
We can now write
\begin{align}
\P&\left( \sum_l S_lS'_l - \frac{m^2a_l^2}{4} > \tau_1 + \tau_2 \right)
\nonumber
\\&\leq \P\left(\xi^c\right) + \P\left(\left. \sum_l \left(S_l - \frac{ma_l}{2}\right)\left(S'_l - \frac{ma_l}{2}\right) > \tau_1\right| \xi  \right) 
\nonumber
\\&\qquad\qquad\quad + \P\left(\left. \sum_l \frac{ma_l}{2}\left(S_l + S'_l - ma_l\right) > \tau_2\right| \xi  \right).
\label{eqn_sum_pos_pf2}
\end{align}
For the first term, we note that due to the Bernstein inequality,
\begin{align*}
\P\left(\xi_l^c\right) = \P\left(\left|S_l - \frac{ma_l}{2}\right| \leq z\right) 
\leq 2\exp\left(\frac{-z^2}{mv_l + \frac43 z}\right)
\end{align*}
noting that $|X_{kl}-a_l| \leq 2$.
Substituting $z$ and noting that $v_l\leq v$, the above bound is at most $\epsilon$. Hence, by union bound $\P\left(\xi^c\right) \leq 4d\epsilon$.
To deal with the other two terms in~\eqref{eqn_sum_pos_pf2}, we need the following claim.
\begin{clm}
The following relations hold for all $l=1,\ldots,d$:
\\(i) $\{S_l,S'_l : l=1,\ldots,d\}$ are mutually independent after conditioning on $\xi$, and
(ii) $\Var(S_l |\xi) = \Var(S_l | \xi_l) \leq \Var(S_l)$.
\end{clm}
\begin{proof} 
Note that without conditioning $\{S_l,S'_l : l=1,\ldots,d\}$ are mutually independent, and $\xi_l$ is defined only on $S_l$. Hence, $\{\xi_l,\xi'_l : l=1,\ldots,d\}$ are independent, and moreover, $\xi_l$ is independent of the mentioned random variables apart from $S_l$. From this observation, (i) follows.

The equality in (ii) follows directly from the above arguments. To prove the inequality, define the non-negative random variable $Y = (S_l - \frac12 ma_l)^2$, and note that $\xi_l = \{Y\leq z^2\}$. Hence,
\begin{align*}
\Var(S_l) = \E[Y] &= \E[Y \1\{Y\leq z^2\}] + \E[Y \1\{Y> z^2\}]
\\&\geq \E[Y| \xi_l] \P(\xi_l) + z^2 \P(\xi_l^c)
\\&= \E[Y|\xi_l] + (z^2 - \E[Y|\xi_l]) \P(\xi_l^c) \geq \E[Y|\xi_l]
\end{align*}
since $\E[Y|\xi_l] = \frac{\E[Y \1\{X\leq z^2\}]}{\P(\xi_l)} \leq \frac{z^2\P(\xi_l)}{\P(\xi_l)} = z^2$. Hence, the claim.
\end{proof}

We now apply Bernstein inequality for the second term in~\eqref{eqn_sum_pos_pf2} to obtain
\begin{align*}
\P&\left(\left. \sum_l \left(S_l - \frac{ma_l}{2}\right)\left(S'_l - \frac{ma_l}{2}\right) > \tau_1\right| \xi  \right) 
\\&\leq \exp\left(\frac{-\tau_1^2}{2\sum_l\Var\left( (S_l - \frac{ma_l}{2})(S'_l - \frac{ma_l}{2}) |\xi\right) + \frac23 z^2\tau_1}\right)
\\&\leq d\epsilon \,,
\end{align*}
where we use the claim to write $\Var\left( (S_l - \frac{ma_l}{2})(S'_l - \frac{ma_l}{2}) |\xi\right)$ is at most $\Var(S_l)\Var(S'_l) = \frac{m^2v_l^2}{4}$, and then substitute the value of $\tau_1$.
The third term in~\eqref{eqn_sum_pos_pf2} can be dealt with similarly as
\begin{align*}
\P&\left(\left. \sum_l \frac{ma_l}{2}\left(S_l + S'_l - ma_l\right) > \tau_2\right| \xi  \right)
\\&\leq \exp\left(\frac{-\tau_2^2}{2\sum_l\frac{m^2a_l^2}{4}\left(\Var(S_l |\xi) + \Var(S'_l |\xi) \right) + \frac13 mz (\max_l a_l) \tau_2}\right)
\\&\leq \exp\left(\frac{-\tau_2^2}{\frac12m^3\sum_l a_l^2 v_l + \frac13 maz \tau_2}\right)
\leq d\epsilon \,,
\end{align*}
In the last step, we take $\sum_l  a_l^2 v_l \leq a^2 (\max_l v_l) \leq a^2v$, and substitute $\tau_2$. 
\end{proof}

\section{Proof of results in Section~5}

\subsection{Proof of Lemma~\ref{thm_optest_LB}} 
We follow the generic technique for deriving lower bounds described earlier, where we need to show that~\eqref{eqn_likeratio} holds for some choice of $\theta_0 \in \Omega_0$ and $\Theta_1\subset \Omega_1$.
Let $p=\frac\cc2$, and $\gamma\in(0,p]$.
We choose $\theta_0=(P,Q)$ such that every off-diagonal entry in $P$ and $Q$ equals $p$.
Let $\Theta_1$ be the collection of all $\theta=(P,Q)$, where $P$ is same as before, but $Q$ is chosen in the following way. For every $v\in\{-1,+1\}^n$, we define $Q=Q_v$ such that
$Q_{ij} = (p + \gamma v_i v_j)$ for every $i\neq j$.
One can see that there are exactly $2^n$ elements in $\Theta_1$, each corresponding to a $v\in\{-1,+1\}^n$.
(To be precise, $\Theta_1$ contains $2^{n-1}$ elements since $v$ and $-v$ result in the same $Q$. But, for convenience, we compute the expectation by counting every model twice and divide by $2^n$).  
Note that $\Theta_1\subset \Omega_1$ if $\Vert P-Q_v\Vert_{op} = \gamma \Vert vv^T-I\Vert_{op} = \gamma(n-1) > \sc$.

We now compute the quantity in~\eqref{eqn_likeratio}. As before, let $\omega\in\mathcal{F}$ correspond to the tuple $\omega = (G_1,\ldots,G_m,H_1,\ldots,H_m)$, where we assume that the first $m$ graphs are generated from the first model, and the rest from the second model.
Then 
\begin{align*}
\P_{\theta_0}(\omega) = \prod_{i<j} p^{(S_G)_{ij} + (S_H)_{ij}}(1-p)^{2m - (S_G)_{ij} - (S_H)_{ij}}\;,
\end{align*}
where $S_G = \sum_k A_{G_k}$ and $S_H = \sum_k A_{H_k}$.
On the other hand, every element in $\Theta_1$ is characterised by $v \in \{\pm1\}^{n}$. 
Denoting the element by $\theta_v$, we have
\begin{align*}
\P_{\theta_v}(\omega) = \prod_{i<j}p^{(S_G)_{ij}} (1-p)^{m -  (S_G)_{ij}}
(p+v_iv_j\gamma)^{(S_H)_{ij}}(1-p-v_iv_j\gamma)^{m - (S_H)_{ij}}\;.
\end{align*}
The quantity in~\eqref{eqn_likeratio} can be computed as 
\begin{align*}
\sum\limits_{\omega\in\mathcal{F}}& \frac{(\E_{\theta_1\sim \text{Unif}(\Theta_1)} [\P_{\theta_1}(\omega)])^2}{\P_{\theta_0}(\omega)}
\\=& \frac{1}{2^{2n}}\sum_\omega\sum_{v,v'} \prod_{i<j}
\frac{p^{(S_G)_{ij}}(1-p)^{m - (S_G)_{ij}}}{p^{(S_H)_{ij}}(1-p)^{m - (S_H)_{ij}}}(p+v_iv_j\gamma)^{(S_H)_{ij}}
\\&\times(1-p-v_iv_j\gamma)^{m - (S_H)_{ij}}
(p+v'_iv'_j\gamma)^{(S_H)_{ij}}(1-p-v'_iv'_j\gamma)^{m - (S_H)_{ij}}
\\=& \frac{1}{2^{2n}} \sum_{v,v'} \prod_{i<j} \left(1 +  \frac{v_iv_jv'_iv'_j\gamma^2}{p(1-p)} \right)^m
\\\leq& \frac{1}{2^{2n}} \sum_{v,v'} \exp \left(\frac{2m\gamma^2}{p} \sum_{i<j} v_iv'_iv_jv'_j\right)
\end{align*}
where the second equality follows steps similar to derivation leading to~\eqref{eqn_frotest_LB_prf1}, and in the last step, we note $p\leq \frac12$.
Note that the above term can be viewed as an expectation of the summation where $v,v'$ are independent and uniformly distributed over $\{\pm1\}^n$.
Note that the function depends only on $z=v\circ v'$, where $\circ$ denote coordinate-wise product, and if $v,v'$ are i.i.d. uniform over $\{\pm1\}^n$, then $z$ is also uniform on $\{\pm1\}^n$. Thus, the above bound may be expressed as $\E_z\left[\exp \left(\frac{2m\gamma^2}{p} \sum_{i<j} z_iz_j\right)\right] = \E_z\left[\exp \left(\frac{m\gamma^2}{p} (S_n^2-n)\right)\right]$ defining $S_k = \sum_{i\leq k} z_i$.

We now claim the following.
\begin{clm}
\label{clm_pf_opLB_1}
Let $(c_l)_{l=0,\ldots,n-1}$ be such that $c_0\leq \frac{1}{32n}$ and $c_{l+1} = c_l + 8c_l^2$.
Then $c_l \leq c_0\left(1+ \frac{l}{n-1}\right) \leq 2c_0$. 
\end{clm}
\begin{proof}
The claim can be proved by induction. If the first bound holds for $c_l$, then
\begin{align*}
c_{l+1} \leq \left( 1 + \frac{l}{n-1}\right)c_0 + 32c_0^2 \leq c_0\left( 1 + \frac{l}{n-1} + \frac{1}{n}\right)
\end{align*}
if $c_0\leq \frac{1}{32n}$, which leads to the desired bound on $c_{l+1}$.
\end{proof}

\begin{clm}
\label{clm_pf_opLB_2}
For any $c\leq 2c_0 \leq \frac{1}{16n}$,
\begin{equation*}
\E_{z_1,\ldots,z_{l+1}}\left[\exp\left(cS_{l+1}^2\right)\right] \leq \exp(c+8c^2)\E_{z_1,\ldots,z_l}\left[\exp\left((c+8c^2)S_l^2\right)\right].
\end{equation*}
\end{clm}
\begin{proof}
Observe that $\exp(x) \leq (1 + x + 2x^2)$ for all $x\in[-1,1]$, and  $|2cS_lz_{l+1} + c| \leq 1$  for any $c\leq\frac{1}{16n}$.
The bound follows since 
\begin{align*}
\E_{z_1,\ldots,z_{l+1}}\left[\exp\left(cS_{l+1}^2\right)\right]
&= \E_{z_1,\ldots,z_{l}}\left[\exp\left(cS_{l}^2\right)\E_{z_{l+1}}\left[\exp\left(2cS_lz_{l+1} + c\right)\right]\right]
\\&\leq \E_{z_1,\ldots,z_{l}}\left[\exp\left(cS_{l}^2\right)\left(1 + c + 2c^2 + 8c^2S_l^2\right)\right]
\\&\leq (1+c+2c^2) \E_{z_1,\ldots,z_{l}}\left[\exp\left((c+8c^2)S_{l}^2\right)\right].
\\&\leq\exp(c+8c^2) \E_{z_1,\ldots,z_{l}}\left[\exp\left((c+8c^2)S_{l}^2\right)\right].
\end{align*}
The first inequality uses the bound on $\exp(x)$ and then takes expectation noting that $z_{l+1}^2 = 1$.
\end{proof}

Setting $c_0 = \frac{m\gamma^2}{p}$ and using these two claims repeatedly, we bound
\begin{align*}
\sum\limits_{\omega\in\mathcal{F}} \frac{(\E_{\theta_1\sim \text{Unif}(\Theta_1)} [\P_{\theta_1}(\omega)])^2}{\P_{\theta_0}(\omega)}
&\leq \exp(-c_0n) \E_{z_1,\ldots,z_n} \left[\exp\left(c_0 S_n^2\right)\right]
\\&\leq \exp(-c_0n+c_1) \E_{z_1,\ldots,z_{n-1}} \left[\exp\left(c_1 S_{n-1}^2\right)\right]
\\&\leq \exp\left(-c_0 n + \sum_{i=1}^{n-1} c_i\right) \E_{z_1}\left[\exp\left(c_{n-1} z_1^2\right)\right]
\\&\leq \exp\left(c_0 n\right) 
\end{align*}
since $c_i\leq 2c_0$ for all $i$.
Note that the upper bound equals $\exp(\frac{mn\gamma^2}{p})$, and holds if $c_0\leq \frac{1}{32n}$, that is $\gamma \leq \sqrt{\frac{p}{32mn}}$.
On the other hand, the inequality in~\eqref{eqn_likeratio} holds if $\gamma \leq \ell_\eta\sqrt{\frac{p}{mn}}$, where $\ell_\eta = \sqrt{\ln(1+4(1-\eta)^2)}$. Hence, the minimax risk is at least $\eta$ if $\gamma \leq (\ell_\eta\wedge \frac{1}{\sqrt{32}})\sqrt{\frac{p}{mn}} = 2\ell'_\eta\sqrt{\frac{\cc}{mn}}$ for $\ell'_\eta$ defined in Theorem~\ref{thm_optest}.
We now set $\gamma = \frac\cc2 \wedge 2\ell'_\eta\sqrt{\frac{\cc}{mn}}$, and obtain the stated claim by substituting this value of $\gamma$ in the condition $\gamma(n-1) \geq \frac{\gamma n}{2} > \sc$.

\subsection{Proof of Lemma~\ref{thm_optest_UB}} 

The general approach for the proof is along the lines of Lemma~\ref{thm_frotest_UB}, but we now need concentration inequalities for operator norm and row sum norm.

\paragraph{\textbf{Preparatory computations and inequalities}}~\\
The concentration result for operator norm that we use is a direct consequence of the matrix Bernstein inequality~\citep{Tropp_2012_jour_FOCM,Oliveira_2009_arxiv_0600}. 
The proof, given at the end of this subsection, is similar to the concentration result for random adjacency matrices in Theorem 3.1 of~\citet{Oliveira_2009_arxiv_0600}.
\begin{lem}[Concentration of operator norm]
\label{lem_matrixBernstein}
Let $\theta=(P,Q)$. For any $\tau$ such that $0<\tau\leq m\Vert P+Q\Vert_{row}$, we have
\begin{displaymath}
\P_\theta\left(\Vert S^- - m(P-Q)\Vert_{op} > \tau\right) \leq 2n\exp\left(\frac{-\tau^2}{3m\Vert P+Q\Vert_{row}}\right).
\end{displaymath}
\end{lem}

We also use the following concentration result for the row sum norm, which is also proved later.
\begin{lem}[Concentration of row sum norm]
\label{lem_row_sum_conc}
For any $\theta=(P,Q)$,
\begin{align*}
\P_\theta\left(\Vert S^+\Vert_{row} \geq 2 m\Vert P+Q\Vert_{row}  \right) &\leq n\exp\left(-\frac{m\Vert P+Q\Vert_{row}}{8}\right),
\\ \text{and}\quad
\P_\theta\left(\Vert S^+\Vert_{row} \leq \frac{m}{4} \Vert P+Q\Vert_{row}  \right) &\leq \exp\left(-\frac{m\Vert P+Q\Vert_{row}}{8}\right) .
\end{align*}
Furthermore, if $\Vert P+Q\Vert_{row} \leq \frac{\tau}{4m}$, then
\begin{align*}
\P_\theta\left(\Vert S^+\Vert_{row} \geq 2\tau  \right) &\leq n\exp\left(-\tau\right).
\end{align*}
\end{lem}

\paragraph{\textbf{Controlling the type-I error rate}}~\\
We now begin the proof of Lemma~\ref{thm_optest_UB}. For $\theta =(P,Q) \in \Omega_0$, observe that $\E[S^-] = 0$ and \begin{align}
\P_\theta&(\Psi_{op}=1) 
\label{eqn_opUB_H0}
\\&\leq \P_\theta\left( \frac{\Vert S^- \Vert_{op}}{\sqrt{\Vert S^+ \Vert_{row}}} > t_1 \sqrt{\ln\left(\frac{n}{\eta}\right)} \right) \bigwedge \P_\theta\left( \Vert S^+\Vert_{row} > t_2\ln\left(\frac{n}{\eta}\right)\right)
\nonumber
\end{align}
As in Section~\ref{Proof_thm_frotest_UB}, we distinguish between two cases.

\begin{description}
\item[\normalfont\underline{\textit{Case 1: $\Vert P+Q\Vert_{row} \geq \frac{t_1^2}{4m}\ln(\frac{n}{\eta})$:}}]~
We can write
\begin{align*}
\P_\theta&\left( \frac{\Vert S^- \Vert_{op}}{\sqrt{\Vert S^+ \Vert_{row}}} > t_1 \sqrt{\ln\left(\frac{n}{\eta}\right)} \right) 
\\&\hskip5ex\leq 
\P_\theta\left( \Vert S^- \Vert_{op} > \frac{t_1}{2} \sqrt{m\Vert P+Q\Vert_{row}\ln\left(\frac{n}{\eta}\right)} \right) 
\\&\hskip25ex+ \P_\theta\left(\Vert S^+\Vert_{row}<\frac{m}{4}\Vert P+Q\Vert_{row}\right)
\end{align*}
Due to Lemma~\ref{lem_matrixBernstein}, the first term is bounded by $2n\exp(-\frac{t_1^2}{12} \ln(\frac{n}{\eta})) \leq 2(\frac{\eta}{n})^{t_1^2/12-1}$.
Using second inequality of Lemma~\ref{lem_row_sum_conc} and the lower bound on $\Vert P+Q\Vert_{row}$, we can say that the second term is bounded by $(\frac{\eta}{n})^{t_1^2/32}$.
Since $\eta<1$ and $n\geq2$ (otherwise, we deal with empty graphs on one vertex), the above sum can be made smaller than $\frac\eta2$ by choosing $t_1$ large enough.

\item[\normalfont\underline{\textit{Case 2: $\Vert P+Q\Vert_{row} \leq \frac{t_1^2}{4m}\ln(\frac{n}{\eta})$:}}]~
In this case, set $t_2 = 2t_1^2$ and use the last bound of Lemma~\ref{lem_row_sum_conc} to conclude that
\begin{align*}
\P_\theta\left( \Vert S^+\Vert_{row} > t_2\ln\left(\frac{n}{\eta}\right)\right) 
\leq n\exp\left(-t_1^2\ln\left(\frac{n}{\eta}\right)\right) \leq \left( \frac{\eta}{n}\right)^{t_1^2-1},
\end{align*}
which is at most $\frac\eta2$ for large enough $t_1$.
Thus, $\P_\theta(\Psi_{op}=1) \leq \frac\eta2$ for all $\theta\in\Omega_0$.
\end{description}

\paragraph{\textbf{Controlling the type-II error rate}}~\\
We now bound the Type-II error rate. For $\theta=(P,Q)\in \Omega_1$,
\begin{align}
\P_\theta&(\Psi_{op}=0) 
\nonumber
\\&\leq \P_\theta\left( \frac{\Vert S^- \Vert_{op}}{\sqrt{\Vert S^+ \Vert_{row}}} \leq t_1 \sqrt{\ln\left(\frac{n}{\eta}\right)} \right) + \P_\theta\left( \Vert S^+\Vert_{row} \leq t_2\ln\left(\frac{n}{\eta}\right)\right)
\nonumber
\\&\leq \P_\theta\left( \Vert S^- \Vert_{op} \leq \frac{3t_1}{2} \sqrt{m\Vert P+Q\Vert_{row}\ln\left(\frac{n}{\eta}\right)} \right) 
\label{eqn_pf_opH1}
\\&~~+ \P_\theta\left(\Vert S^+\Vert_{row}\geq \frac{9m}{4}\Vert P+Q\Vert_{row}\right)
+ \P_\theta\left( \Vert S^+\Vert_{row} \leq 2t_1^2\ln\left(\frac{n}{\eta}\right)\right)
\nonumber
\end{align}
Before bounding the individual terms, we recall that the sufficient condition on $\sc$ implies that
\begin{align}
\Vert P+Q\Vert_{row} \geq \Vert P-Q\Vert_{op} > \sc \geq \frac{C'}{m}\ln\left(\frac{n}{\eta}\right)
\label{eqn_pf_oprow}
\end{align}
For $C'\geq 8t_1^2$, we can bound the third term in~\eqref{eqn_pf_opH1} by
\begin{align*}
\P_\theta\left(\Vert S^+\Vert_{row} \leq \frac{m}{4} \Vert P+Q\Vert_{row}  \right) &\leq \exp\left(-\frac{m\Vert P+Q\Vert_{row}}{8}\right) \leq \left(\frac{\eta}{n}\right)^{C'/8}.
\end{align*}
The second term is similarly bounded by
\begin{align*}
\P_\theta\left(\Vert S^+\Vert_{row} \geq 2m \Vert P+Q\Vert_{row}  \right) &\leq n\exp\left(-\frac{m\Vert P+Q\Vert_{row}}{8}\right) \leq \left(\frac{\eta}{n}\right)^{C'/8-1}.
\end{align*}

To bound the first term in~\eqref{eqn_pf_opH1}, we use $\Vert P+Q\Vert_{op} \leq n\cc$ and the sufficient condition on $\sc$ to see that
\begin{align*}
\frac{3t_1}{2} \sqrt{m\Vert P+Q\Vert_{row}\ln\left(\frac{n}{\eta}\right)}
\leq \frac{3t_1}{2} \sqrt{mn\cc\ln\left(\frac{n}{\eta}\right)}
\leq \frac{3t_1}{2C} m\sc
< \frac{m}{2}\Vert P-Q\Vert_{op}
\end{align*}
where the last inequality holds for large $C$.
Hence, the first term in~\eqref{eqn_pf_opH1} is bounded by
\begin{align*}
&\P_\theta\left( \Vert S^- \Vert_{op} <  \frac{m}{2}\Vert P-Q\Vert_{op}\right) 
\leq \P_\theta\left( \Vert S^- - m(P-Q) \Vert_{op} >  \frac{m}{2}\Vert P-Q\Vert_{op}\right) 
\\&\hskip10ex\leq
\P_\theta\left( \Vert S^- - m(P-Q)\Vert_{op} > \frac{3t_1}{2} \sqrt{m\Vert P+Q\Vert_{row}\ln\left(\frac{n}{\eta}\right)} \right) 
\end{align*}
The first inequality follows from use of reverse triangle inequality, that is, $\Vert S^- - m(P-Q)\Vert_{op} \geq m\Vert P-Q\Vert_{op} - \Vert S^-\Vert_{op}$.
Now, for large $C'$, we can use~\eqref{eqn_pf_oprow} to claim that $\frac{3t_1}{2} \sqrt{m\Vert P+Q\Vert_{row}\ln\left(\frac{n}{\eta}\right)} \leq m\Vert P+Q\Vert_{row}$.
Hence, Lemma~\ref{lem_matrixBernstein} shows that the probability is bounded by $2(\frac{\eta}{n})^{3t_1^2/4-1}$.
Combining the above three bounds, we argue that the Type-II error rate is smaller than $\frac\eta2$ if $t_1,C,C'$ are large enough. Hence, the maximum risk is at most $\eta$ under the stated sufficient condition.

We conclude this section with the proofs for Lemmas~\ref{lem_matrixBernstein}--\ref{lem_row_sum_conc}.

\begin{proof}[Proof of Lemma~\ref{lem_matrixBernstein}]
Let $\e_1,\ldots,\e_n$ denote the standard basis for $\R^n$. Then we can write $S^-- m(P-Q)$ as
\begin{align*}
S^--m(P-Q) = \sum_{i<j}\sum_{k=1}^m &\left((A_{G_k})_{ij} - P_{ij}\right) \left(\e_i\e_j^T + \e_j\e_i^T\right) 
\\&-  \sum_{i<j}\sum_{k=1}^m \left((A_{H_k})_{ij} - Q_{ij}\right) \left(\e_i\e_j^T + \e_j\e_i^T\right)\;,
\end{align*}
which is a sum of $2m\binom{n}{2}$ independent random matrices. One can see that each of these matrices has zero mean, and its operator norm is bounded by 1 almost surely.
Moreover, for each matrix, we can write
\begin{align*}
\E_\theta \left[\left((A_{G_k})_{ij} - P_{ij}\right)^2 \left(\e_i\e_j^T + \e_j\e_i^T\right)^2\right] = 
P_{ij}(1-P_{ij}) \left(\e_i\e_i^T + \e_j\e_j^T\right)\;.
\end{align*}
Hence, the sum of all such expected matrices is a diagonal matrix with maximum diagonal entry bounded by $m\Vert P+Q\Vert_{row}$. Based on these observations, we use the matrix Bernstein inequality (Theorem~1.4 of~\citet{Tropp_2012_jour_FOCM} or Corollary~7.1 of~\citet{Oliveira_2009_arxiv_0600}) to conclude that
\begin{align*}
\P_\theta\left(\Vert S^- - m(P-Q)\Vert_{op} > \tau\right) \leq 2n\exp\left(\frac{-\tau^2}{2m\Vert P+Q\Vert_{row} + \frac23\tau}\right).
\end{align*}
The claim follows by using the condition $\tau\leq m\Vert P+Q\Vert_{row}$.
\end{proof}

\begin{proof}[Proof of Lemma~\ref{lem_row_sum_conc}]
Let $d_i = \sum_j (P_{ij}+Q_{ij})$ and without loss of generality, assume that the first row sum is largest, that is, $d_1 = \Vert P+Q\Vert_{row}$.
To prove the first inequality, we write
\begin{align*}
\P_\theta&\left(\Vert S^+\Vert_{row} \geq 2 m\Vert P+Q\Vert_{row}  \right)
= \P_\theta\left(\max_i \sum_{j=1}^n S^+_{ij} \geq 2m \Vert P+Q\Vert_{row}  \right)  
\\&\hskip12ex\leq 
\sum_{i=1}^n \P_\theta\left(\sum_{j=1}^n\sum_{k=1}^m (A_{G_k})_{ij} + (A_{H_k})_{ij} \geq 2m \Vert P+Q\Vert_{row}  \right)
\end{align*}
using union bound. The probability corresponds to the tail of the sum of $2nm$ independent random variables, each lying in the interval $[0,1]$.
Moreover, for any $i$, $\Var\left(\sum_{j,k} (A_{G_k})_{ij} + (A_{H_k})_{ij}\right) \leq md_i$.
Now, consider $i$ such that $d_i \geq \frac13\Vert P+Q\Vert_{row}$. We can use Bernstein inequality to write
\begin{align*}
\P_\theta&\left(\sum_{j=1}^n\sum_{k=1}^m (A_{G_k})_{ij} + (A_{H_k})_{ij} \geq 2m \Vert P+Q\Vert_{row}  \right)
\\&\leq \P_\theta\left(\sum_{j=1}^n\sum_{k=1}^m (A_{G_k})_{ij} - P_{ij} + (A_{H_k})_{ij} - Q_{ij}\geq m d_i \right)
\\&\leq \exp\left(-\frac{m^2d_i^2}{2md_i + \frac23md_i}\right)
\\&\leq \exp\left(-\frac{m\Vert P+Q\Vert_{row}}{8}\right)
\end{align*} 
since $d_i \geq \frac13\Vert P+Q\Vert_{row}$. For other rows, where $d_i < \frac13\Vert P+Q\Vert_{row}$, we have by the Markov inequality
\begin{align*}
\P_\theta&\left(\sum_{j=1}^n\sum_{k=1}^m (A_{G_k})_{ij} + (A_{H_k})_{ij} \geq 2m \Vert P+Q\Vert_{row})  \right)
\\&\leq  \exp(-2m\Vert P+Q\Vert_{row})\prod_{j=1}^n\prod_{k=1}^m \E_\theta[\exp((A_{G_k})_{ij})]\E_\theta[\exp((A_{H_k})_{ij})]
\\&= \exp(-2m\Vert P+Q\Vert_{row})\prod_{j=1}^n \prod_{k=1}^m \left(1+(e-1)P_{ij}\right) \left(1+(e-1)Q_{ij}\right) .
\end{align*}
Using the facts that $P_{ij}Q_{ij} \leq (P_{ij}+Q_{ij})/2$ and  $d_i < \frac13\Vert P+Q\Vert_{row}$, we have
\begin{align*}
\P_\theta&\left(\sum_{j=1}^n\sum_{k=1}^m (A_{G_k})_{ij} + (A_{H_k})_{ij} \geq 2m \Vert P+Q\Vert_{row})  \right)
\\&\leq \exp(-2m\Vert P+Q\Vert_{row})\prod_{j=1}^n \prod_{k=1}^m (1+4(P_{ij}+Q_{ij}))
\\&\leq \exp(-2m\Vert P+Q\Vert_{row} + 4md_i) \leq \exp(-2m\Vert P+Q\Vert_{row}/3).
\end{align*} 
Combining above bounds, we obtain the first inequality in Lemma~\ref{lem_row_sum_conc}.
We prove the second inequality by observing that
\begin{align*}
\P_\theta\left(\Vert S^+\Vert_{row} \leq \frac{m}{4} \Vert P+Q\Vert_{row} \right) 
&= \P_\theta\left(\max_i \sum_{j=1}^n S^+_{ij}  \leq \frac{m}{4} \Vert P+Q\Vert_{row}  \right)  
\\&\leq  \P_\theta\left(\sum_{j=1}^n\sum_{k=1}^m (A_{G_k})_{1j} + (A_{H_k})_{1j} \leq \frac{md_1}{4} \right)
\\&\leq \exp\left(-\frac{(\frac34 md_1)^2}{2md_1 + \frac12md_1}\right)
\\&\leq \exp\left(-\frac{m\Vert P+Q\Vert_{row}}{8}\right)
\end{align*}
where the third line is due to Bernstein inequality, and the last line uses the fact that $d_1=\Vert P+Q\Vert_{row}$. 

The third inequality in Lemma~\ref{lem_row_sum_conc} is proved using Markov inequality in the following way.
\begin{align*}
\P_\theta\big(\Vert S^+\Vert_{row} &\geq 2\tau \big)
\leq \sum_{i=1}^n \P_\theta\left(\sum_{j=1}^n\sum_{k=1}^m (A_{G_k})_{ij} + (A_{H_k})_{ij} \geq 2\tau  \right)
\\&\leq \exp(-2\tau) \sum_{i=1}^n \prod_{j=1}^n\prod_{k=1}^m \E_\theta[\exp((A_{G_k})_{ij})]\E_\theta[\exp((A_{H_k})_{ij})]
\\&\leq \exp(-2\tau)\sum_{i=1}^n\prod_{j=1}^n \prod_{k=1}^m (1+4(P_{ij}+Q_{ij}))
\\&\leq n\exp(- 2\tau + 4m\Vert P+Q\Vert_{row}) \,,
\end{align*}
which is smaller that $n\exp(-\tau)$ for $\tau \geq 4m\Vert P+Q\Vert_{row}$. 
\end{proof}

\subsection{Proof of Proposition~\ref{thm_optest_UB2}} 
We note that our proof is valid for any $n\geq 2$. If one additionally assumes that $n$ is large enough, then the terms involving $\ln(\frac2\eta)$ in the test $\Psi'_{op}$~\eqref{eqn_op_test2} and Proposition~\ref{thm_optest_UB2} can be replaced by absolute constants.

\paragraph{\textbf{Preliminary computations and inequalities}}~\\
We first state the concentration result of~\citet[][Theorem 2.1]{Le_2015_arxiv} adapted to our setting. The adaptation is explicitly described later in the proof of the lemma.
\begin{lem}[Concentration of trimmed adjacency matrix]
\label{lem_trimconc}
Let $G\sim\IER(P)$ with $\Vert P\Vert_{max}\leq \cc$, where $\cc > \frac{10}{n}$.
Consider the trimming procedure which isolates all vertices with degree larger than $cn\cc\ln(\frac2\eta)$, and $A'_G$ be the resulting adjacency matrix.
There exists an absolute constant $C>0$ such that with probability $1 - \frac\eta4$,
\begin{displaymath}
 \Vert A'_G - P\Vert_{op} \leq C\sqrt{n\cc}\ln^2\left(\frac2\eta\right).
\end{displaymath}
\end{lem}

\paragraph{\textbf{Controlling the type-I error rate}}~\\
The proof of Proposition~\ref{thm_optest_UB2} follows directly from the above lemma by setting $t=2C$.
Let $\theta = (P,Q)\in \Omega_0$. 
\begin{align*}
\P_\theta\left(\Psi'_{op} = 1\right)
&=\P_\theta \left( \Vert A'_G - A'_H \Vert_{op} > t\sqrt{n\cc}\ln^2\left(\frac2\eta\right)\right)
\\&\leq\P_\theta \left( \Vert A'_G - P\Vert_{op} + \Vert A'_H - Q \Vert_{op} > 2C\sqrt{n\cc}\ln^2\left(\frac2\eta\right)\right)
\end{align*}
where the bound follows from triangle inequality and noting that $P=Q$.
Lemma~\ref{lem_trimconc} states that each of the two norms can exceed $C\sqrt{n\cc}\ln^2(\frac2\eta)$ with probability at most $\frac\eta4$. Hence, the above probability is bounded by $\frac\eta2$.

\paragraph{\textbf{Controlling the type-II error rate}}~\\
Now consider $\theta=(P,Q)\in\Omega_1$ and use the condition $\sc\geq 4C\sqrt{n\cc}\ln^2(\frac2\eta)$.
By reverse triangle inequality
\begin{align*}
\Vert A'_G - A'_H\Vert_{op} &\geq \Vert P-Q \Vert_{op} - \Vert A'_G - P\Vert_{op} - \Vert A'_H - Q \Vert_{op}
\\&> 4C\sqrt{n\cc}\ln^2\left(\frac2\eta\right) -  \Vert A'_G - P\Vert_{op} - \Vert A'_H - Q \Vert_{op}\,.
\end{align*}
Using the above lower bound for $\Vert A'_G - A'_H\Vert_{op}$, we get
\begin{align*}
\P_\theta\left(\Psi'_{op} = 0\right)
&=\P_\theta \left( \Vert A'_G - A'_H \Vert_{op} \leq 2C\sqrt{n\cc}\ln^2\left(\frac2\eta\right)\right)
\\&\leq\P_\theta \left( \Vert A'_G - P\Vert_{op} + \Vert A'_H - Q \Vert_{op} > 2C\sqrt{n\cc}\ln^2\left(\frac2\eta\right)\right)
\end{align*}
which is also at most $\frac\eta2$ due to Lemma~\ref{lem_trimconc}. Hence, Proposition~\ref{thm_optest_UB2} holds.
We now prove Lemma~\ref{lem_trimconc}.

\begin{proof}[Proof of Lemma~\ref{lem_trimconc}]
Define $d = n\cc$. 
\citet[][Theorem 2.1]{Le_2015_arxiv} holds for the following general regularisation process to obtain $A'_G$.
Consider any subset of at most $\frac{10n}{d}$ vertices and reduce the weights of edges incident on them in an arbitrary way. If $d'$ is the maximum degree in $A'_G$, then for any $r\geq1$,
\begin{align*}
\Vert A'_G - P\Vert_{op} \leq C'r^{3/2}\left(\sqrt{d} + \sqrt{d'}\right)
\end{align*}
with probability at least $1-n^{-r}$, where $C'$ is a constant.

Let us fix $r = 6\ln(\frac2\eta)$. For any $n\geq2$, $n^{-r} \leq \exp(-3\ln(\frac2\eta)) \leq \frac\eta8$.
We claim that if $c\geq5$, then with probability at least $1-\frac\eta8$, $G$ has at most $\frac{10n}{d} = \frac{10}{\cc}$ vertices with degree larger than $cn\cc\ln(\frac2\eta)$.
Hence, after deleting all edges incident on them, $d' =cn\cc\ln(\frac2\eta)$ and the above bound suggests that with probability $1-\frac\eta4$
\begin{align*}
\Vert A'_G - P\Vert_{op} 
&\leq C'\left(6\ln\left(\frac2\eta\right)\right)^{3/2}\left(\sqrt{n\cc} + \sqrt{cn\cc\ln\left(\frac2\eta\right)}\right)
\\&\leq C\sqrt{n\cc}\ln^2\left(\frac2\eta\right)
\end{align*}
for an appropriately defined $C$.

We conclude the proof by showing that at most $\frac{10}{\cc}$ vertices can have degree larger than $cn\cc\ln(\frac2\eta)$. Let $c' = \frac12 c\ln(\frac2\eta)$.
Consider any vertex set $V_1$ of size $n_1$. The probability that every vertex in $V_1$ has degree larger than $2c'n\cc$ is bounded by
\begin{align*}
\P&\left(\sum_{i\in V_1} \sum_j (A_G)_{ij} > 2c'n\cc n_1\right)
\\&\leq \P\left(\sum_{i\in V_1,j\notin V_1} (A_G)_{ij} + \sum_{i,j\in V_1, i<j} (A_G)_{ij}  > c'n\cc n_1\right)
\\&\leq \P\left(\sum_{i\in V_1,j\notin V_1} \left((A_G)_{ij} - P_{ij}\right) + \sum_{i,j\in V_1, i<j} \left((A_G)_{ij} -P_{ij}\right)  > (c'-1)n\cc n_1\right)
\\&\leq \exp\left( - \frac{c'n\cc n_1}{4}\right)  
\end{align*}
where the last step is due to Bernstein inequality and holds for $c'\geq3$.

Now if more than $\frac{10}{\cc}$ vertices have degree larger than $2c'n\cc$, then one can find such a vertex set $V_1$ of size $n_1\in[\frac{10}{\cc}, n]$. But the probability that there exists such a set of size $n_1$ is smaller than
\begin{align*}
\binom{n}{n_1} \exp\left( - \frac{c'n\cc n_1}{4}\right) 
&\leq \exp \left( n_1 \ln\left(\frac{en}{n_1}\right) - \frac{c'}{4}nn_1\cc\right) 
\\&\leq \exp \left( n_1 \left(\ln\left(\frac{en\cc}{10}\right) - \frac{c'}{4}n\cc\right)\right) 
\end{align*}
Recall our assumption $n\cc>10$, and observe that for $x>10$, $\ln(ex/10) \leq \frac{c'}{8}x$. 
Hence, the above probability is smaller than
\begin{align*}
\exp\left(- \frac{c'}{8}n_1 n\cc\right) \leq \exp\left(- \frac{5c}{8}n\ln\left(\frac2\eta\right)\right) \leq \frac\eta8
\end{align*}
for $c\geq6$ and $n\geq2$, which completes the proof.
\end{proof}

\bibliographystyle{imsart-nameyear}
\bibliography{refs}

\end{document}